\PassOptionsToPackage{nameinlink}{cleveref}
\documentclass[a4paper,UKenglish,cleveref,thm-restate]{lipics-v2021}
\usepackage[subrefformat=simple,labelformat=simple,position=b]{subcaption}

\usepackage{nicefrac}
\usepackage{mathtools}
\usepackage{xspace}
\usepackage{complexity}
\usepackage{xcolor}
\usepackage{makecell}
\usepackage{todonotes}

\usepackage{etoolbox}

\newif\ifrestatingthm
\restatingthmfalse

\DeclareRobustCommand{\thmstar}[2]{%
	\ifrestatingthm
	\hyperref[#1]{\texorpdfstring{$\star$}{*}}%
	\else
	\hyperlink{#2}{\texorpdfstring{$\star$}{*}}%
	\fi
}

\crefname{figure}{Figure}{Figures}
\crefname{theorem}{Theorem}{Theorems}
\crefname{lemma}{Lemma}{Lemmas}
\crefname{corollary}{Corollary}{Corollaries}
\crefname{section}{Section}{Sections}
\crefname{appendix}{Appendix}{Appendices}
\crefname{remark}{Remark}{Remarks}
\crefname{claim}{Claim}{Claims}
\crefname{conjecture}{Conjecture}{Conjectures}
\crefname{observation}{Observation}{Observations}

\graphicspath{{./figures/}}
\title{Restless Guards: Robust Guarding of Polyominoes}

\titlerunning{Restless Guards}

\author{Linda Kleist}{Department of Computer Science, University of Hamburg}{linda.kleist@uni-hamburg.de}{https://orcid.org/0000-0002-3786-916X}{}
\author{Ramin Kosfeld}{Department of Computer Science, TU Braunschweig}{kosfeld@ibr.cs.tu-bs.de}{https://orcid.org/0000-0002-1081-2454}{}
\author{Christian Rieck}{Institute of Mathematics, University of Kassel}{christian.rieck@mathematik.uni-kassel.de}{https://orcid.org/0000-0003-0846-5163}{}

\authorrunning{Linda Kleist, Ramin Kosfeld, and Christian Rieck}

\Copyright{Linda Kleist, Ramin Kosfeld, and Christian Rieck}



\DeclarePairedDelimiter\floor{\lfloor}{\rfloor}

\newtoggle{insectionheading}
\newcommand{\head}[1]{\toggletrue{insectionheading}#1\togglefalse{insectionheading}}
\let\oldsection\section
\renewcommand{\section}[1]{\head{\oldsection{\boldmath{#1}}}}
\let\oldsubsection\subsection
\renewcommand{\subsection}[1]{\head{\oldsubsection{\boldmath{#1}}}}
\newcommand{\ifsectionheading}[2]{\iftoggle{insectionheading}{#1}{#2}}

\newcommand{\compatible}[1]{\texorpdfstring{\ifsectionheading{\boldmath{$#1$}-}{$#1$\nobreakdash-}}{#1-}}

\newcommand{\rguard}[1]{\compatible{#1}restless guard\xspace}

\newcommand{\rguards}[1]{\compatible{#1}restless guards\xspace}

\newcommand{\guarding}[1]{\compatible{#1}restless guarding\xspace}

\newcommand{\guard}[1]{\compatible{#1}restless guard\xspace}
\newcommand{\guarded}[1]{\compatible{#1}restless guarded\xspace}

\newcommand{\guardable}[1]{\compatible{#1}restless guardable\xspace}    

\newcommand{\outerdent}{outer dent\xspace}
\newcommand{\outerdents}{outer dents\xspace}

\newcommand{\vregion}{visibility region\xspace}
\newcommand{\vregions}{visibility regions\xspace}

\newcommand{\Vregions}{Visibility regions\xspace}

\newcommand{\bigO}{\mathcal{O}\xspace}
\newcommand{\bO}{\bigO}

\newcommand{\emptybase}{bordered base\xspace}
\newcommand{\fullbase}{inside base\xspace}

\newcommand{\block}{\ifmmode B\else$B$\xspace\fi}

\DeclareMathOperator{\vis}{Vis}

\newcommand{\viska}[2]{\ensuremath{\vis_{#1}(#2)}}
\DeclareMathOperator{\opt}{\textsc{opt}\xspace}
\DeclareMathOperator{\apx}{\mathcal{A}}
\DeclareMathOperator{\rect}{rect}
\DeclareMathOperator{\minguard}{\textsc{min}\xspace}
\DeclareMathOperator{\offs}{\textsc{offset}\xspace}

\DeclareMathAlphabet{\mathsfit}{T1}{\sfdefault}{\mddefault}{\sldefault}
\hideLIPIcs
\nolinenumbers

\begin{document}
\maketitle

\begin{abstract}
The art gallery problem is a fundamental problem whose practical applications have inspired numerous variants. 
In this work, we investigate a setting in which full coverage of the gallery shall be ensured despite locational uncertainty of the guards.

We study this problem for polyominoes under rectangular visibility, where a \rguard{k} assigned to a cell may move to any cell within distance at most $k$.
While all polyominoes are guardable for $k\in\{0, 1\}$, the same is not true for all $k\geq 2$.
However, for a given polyomino $P$ and integer $k$, we present an algorithm that identifies a \rguard{k} set of $P$ in polynomial time, or decides that no such set exists. 
For the cases $k\in\{0,1\}$, where \rguard{k} sets always exist, we provide tight worst-case bounds on the number of \rguards{k} that are always sufficient and sometimes necessary to guard a hole-free polyomino with $n$ vertices or $m$ cells.
 	
Moreover, we show that computing a minimum guard set is \NP-hard.
Complementing the hardness, we provide polynomial-time algorithms for solving the problem optimally for histograms, and leverage this to obtain a $20k$-approximation for polyominoes without holes.
\end{abstract}


\section{Introduction}
\label{sec:introduction}

The Art Gallery Problem is a classic and fundamental problem in computational geometry:
Given a polygon $P$ in the plane, what is the minimum set $G$ of points (called \emph{guards}) such that every point of $P$ is \emph{seen} by at least one guard in $G$.
The original question, that gave rise to this area of research more than 50 years ago, is due to Klee~(see~\cite{o1987art}) and was concerned with establishing bounds on the number of guards that are always sufficient and sometimes necessary to guard a hole-free polygon with $n$ vertices.
Chvátal~\cite{chvatal1975combinatorial} proved that the answer is~$\floor{\nicefrac{n}{3}}$. Subsequently, an elegant and simple proof was given by Fisk~\cite{fisk1978short}.

Over the years, many variants of the problem have been investigated, based on constraints and properties such as the allowed guard positions (vertices, edges, interior), the shape of the art gallery (polygons with or without holes, orthogonal, monotone, polyominoes), the parts to be guarded (vertices, boundary, interior), as well as the type of vision (straight-line visibility, rectangular visibility, omni-directional field-of-view, restricted vision~cone).

Recent research also studies variants that incorporate more realistic assumptions, such as more accurate lighting conditions~\cite{eisenbrand2005energy,openproblemsCCCG} or constraints on visibility~\cite{danner2000randomized,ernestus2017algorithms,kazazakis2002fast}.
In this vein, we explore a scenario with locational uncertainty.
In an optimal solution for the classic art gallery problem, guards might be required to be at points with irrational coordinates~\cite{irrationalGuards}.
However, in real-world scenarios, it is unrealistic to assume that guards remain (or are installed) precisely at predetermined positions; instead, they are typically located in a small vicinity of their intended position. 
This requires developing models and algorithms that remain robust to placement inaccuracies while still ensuring full coverage.
Most recently, Das, Filtser, Katz, and Mitchell~\cite{robustly-guarding-polygons-socg} introduced the concept of robust visibility, where a guard's allowed movement depends on the distance to the point being guarded.

In this work, we investigate guard sets that accommodate such inaccuracies (regardless of the distance between a guard and a point), while still ensuring that the polygon is guarded. 
In particular, we investigate the case where the art gallery is modeled as a polyomino and guards can move within a $k$-neighborhood; 
we consider the guards to be \emph{restless}. 
Specifically, a \rguard{k} is placed in some cell of the polyomino and may move to any cell within distance at most $k$; the distance is measured in terms of the number of edges of the dual graph of the polyomino. 
A set of points $G$ in a polyomino is a \emph{\rguard{k} set} if, regardless of the guards' exact positions within their respective $k$-neighborhoods, the entire polyomino is guarded.
We use the standard visibility model for rectangular domains, where a point $p$ is visible from a guard $g$ whenever the axis-aligned rectangle defined by $p$ and $g$ is entirely contained in $P$. 
Thus, in our setting, the visibility region of a \rguard{k} is given by the intersection of the visibility regions from all reachable positions; for examples consider \cref{fig:intro-examples}.
For precise definitions, we refer to~\cref{sec:preliminaries}.

\begin{figure}[htb]
	\centering
	\captionsetup[subfigure]{justification=centering}
	\begin{subfigure}{0.3\columnwidth}%
		\centering
		\includegraphics[page=10]{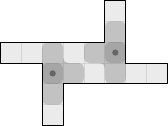}
		\subcaption{}
		\label{fig:intro-examples-a}
	\end{subfigure}%
	\begin{subfigure}{0.3\columnwidth}%
		\centering
		\includegraphics[page=11]{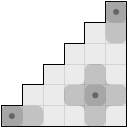}
		\subcaption{}
		\label{fig:intro-examples-b}
	\end{subfigure}%
	\begin{subfigure}{0.4\columnwidth}%
		\centering
		\includegraphics[page=9]{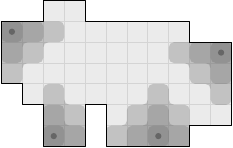}
		\subcaption{}
		\label{fig:intro-examples-c}
	\end{subfigure}%
	\caption{
		Are the depicted \rguard{k} sets valid or minimal for small values of $k$? 
		How many guards appear in optimal solutions? Solutions are given in the footnote.\footnotemark{}
	}
	\label{fig:intro-examples}
\end{figure}
	\footnotetext{
		(a) Optimal for $k=0$, not valid for $k=1$. 
		(b) Valid but not minimal for $k=0$, not valid for $k=1$.\\
		(c) Valid for $k=0,1,2$. The minimum numbers of guards are 2 ($k=0$), and 3 ($k=1,2$).
	}

With these notions at hand, our problem reads as follows.

\subparagraph*{The Art Gallery Problem with Restless Guards.}
Given a polyomino~$P$ and an integer $k \geq 0$, what is the size of a smallest \rguard{k} set of $P$?

\subsection{Our contribution}
We provide both combinatorial and algorithmic results for guarding polyominoes with a set of \rguards{k} for different values of $k$.

\subparagraph*{Guardability and tight worst-case bounds.}
We start by investigating the existence of \rguard{k} sets. While all polyominoes are guardable for $k\in\{0,1\}$, the same is not true for all $k\geq 2$ (\cref{prop:existence}).
However, we present a polynomial-time algorithm that decides for a given polyomino with $n$ vertices and an integer $k$, whether a \rguard{k} set exists (\cref{prop:existencePolytime}). 
We highlight that this algorithm (and all of our other algorithms) is polynomial in the encoding of the input; in particular, we encode a polyomino as a sequence of vertex coordinates rather than a set of cells.

For the two cases where \rguard{k} sets exist, we provide tight worst-case bounds on their size, depending on the number of vertices and cells of the polyomino, denoted by~$n$ and~$m$, respectively. 
In particular, for \rguards{0}, we show that $\floor{\nicefrac{n}{4}}$ as well as $\floor{\nicefrac{m}{2}}$ guards (for $m\geq 2$) are always sufficient, and sometimes necessary (\cref{prop:0-restless-guards-can-guard-polyominoes}). 
In case of \rguards{1}, $\nicefrac{(n-2)}{2}$ as well as $m-2$ guards are always sufficient and sometimes necessary, when $m\geq 4$ (\cref{prop:1-restless-guards-can-guard-polyominoes}).

\subparagraph*{Computing small guard sets: Hardness and approximation.}

As the existence of \rguard{k} sets is decidable in polynomial time, we are interested in the computational complexity of computing optimal or small guard sets.
To this end, we show that the problem is \NP-hard for every fixed $k\in \mathbb{N}$ by a reduction from \textsc{Vertex Cover} (\cref{prop:np-hardness-for-any-k}).

To complement this,
we present a $20k$-approximation algorithm for computing \rguard{k} sets in hole-free
polyominoes (\cref{thm:k-guard-poly}), if they exist. 
As a stepping stone, we show how to compute minimum \rguard{k} sets for histograms (\cref{thm:k-guard-opt-hist}).
We~then turn towards the more general case of
so-called \emph{embedded} histograms and  present a $4k$\nobreakdash-approximation (\cref{lem:k-guard-embedded}); 
as the boundaries are not necessarily present in this setting, computing guard sets is a much  more complex task. 
Lastly, for $k\in\{0,1\}$, we establish improved approximation guarantees, namely $2$ for $k=0$ (\cref{prop:0-approx-2}), and $(10+\varepsilon)$ for $k=1$~(\cref{thm:1-guard-poly}).

\medskip
Full details of the proofs for statements marked with ($\star$) are given in the appendix.

\subsection{Related work}
Since its introduction, the art gallery problem has inspired extensive research, resulting in hundreds of publications on numerous variants. 
We provide a concise overview of the most essential related work, and refer to the book by O'Rourke~\cite{o1987art} and surveys by Shermer~\cite{shermer1992recent} and Urrutia~\cite{Urrutia00} for a comprehensive introduction.
One important aspect concerns the algorithmic and complexity-theoretic perspective.
Abrahamsen, Adamaszek, and Miltzow prove that the art gallery problem is $\exists \mathbb{R}$-complete~\cite{agp-exist-r-complete}, and almost all more restricted variants remain \NP-hard, e.g., for vertex guards, even in polygons without holes~\cite{LeeL86}, orthogonal polygons~\cite{SchuchardtH95}, or polyominoes~\cite{BiedlIIKM11}.
Stade~\cite{Stade25} shows that the problem remains $\exists \mathbb{R}$-complete even when only the polygon boundary shall be guarded.
In contrast, when each guard is required to cover exactly one contiguous portion of the boundary the problem can be solved in polynomial time~\cite{BiniazMMMOPRRSS25,deberg2025contiguousartgalleryproblem}. 
Concerning approximation algorithms, Eidenbenz, Stamm, and Widmayer~\cite{EidenbenzSW01} show that there exists no $o(\log(n))$-approximation for polygons with holes, unless \P=\NP. 
Bonnet and Miltzow~\cite{BonnetM17} provide a randomized $\mathcal{O}(\log(n))$-approximation algorithm assuming vertices at integer coordinates 
and strong general position. 
For a few polygon classes such as weakly-visible polygons~\cite{AshurFK21,AshurFKS22}, terrains~\cite{Ben-MosheKM07}, or monotone polygons~\cite{Nilsson05}, constant factor approximations or even {\sffamily PTAS}s are known. 
In the remainder, we focus on the most relevant work related to our paper, namely polyominoes and guarding with imprecision. 

\subparagraph{Polyominoes.}
Orthogonal polygons are generalizations of polyominoes, for which the worst-case bound of $\floor{\nicefrac{n}{3}}$ can be improved.
Kahn, Klawe, and Kleitman~\cite{kahn1983traditional} show that $\floor{\nicefrac{n}{4}}$ guards are always sufficient and sometimes necessary to guard an orthogonal polygon with $n$ vertices under straight-line visibility; 
O'Rourke~\cite{o1983alternate} gives a simpler proof. 

Restricting the input polygons  to polyominoes, Biedl, Irfan, Iwerks, Kim, and Mitchell~\cite{BiedlIIKM11} show that the problem remains \NP-hard for hole-free polyominoes with straight-line visibility. 
Furthermore, the same set of authors~\cite{agpForPolynominoes} prove that, in terms of the number $m$ of cells, $\floor{\nicefrac{(m+1)}{3}}$ guards are always sufficient and sometimes necessary to guard a polyomino, both with straight-line visibility and rectangular visibility.
Moreover, if entire cells can be chosen as guards (so-called \emph{pixel guards}), they provide a bound of $\floor{\nicefrac{3m}{11}}+1$ for both visibilities~\cite{agpForPolynominoes}. 
We highlight that pixel guards are closely related to but more powerful than \rguards{0}, as the visibility of a pixel guard is the union of the visibility regions of all points in the cell rather than the intersection.

In contrast, for hole-free orthogonal polygons under rectangular visibility, Worman and Keil~\cite{worman07} show that the problem becomes polynomial-time solvable. 
The algorithm computes a maximum independent set of a perfect graph as a subroutine and has a running time in~$\mathcal{\tilde O}(n^{17})$.
Motivated by the high running time, Lingas, Wasylewicz, and \.Zyli\'nski~\cite{3approx} provide a linear-time $3$-approximation algorithm for the same setting.
Biedl and Mehrabi~\cite{Biedl016,BiedlM21} prove \NP-completeness, even for thin orthogonal polygons possibly with holes and rectangular visibility. 
Furthermore, they provide a linear-time algorithm for thin orthogonal polygons without holes, and argue that the problem becomes linear-time solvable for hole-free orthogonal polygons with bounded treewidth. 
Their methods actually generalize to orthogonal polygons with $h$ holes and thickness $t$, i.e., the problem is fixed-parameter tractable in $h+t$.

Guarding polyominoes was also considered under the notion of $k$-hop visibility in which each guard's vision is limited to paths of length~$k$ in the dual graph of the polyomino~\cite{k-hop-visibility-polyominoes-algorithmica}.
With this notion, guards may be able to guard points that lie behind corners.
The problem is \NP-complete for any fixed~$k$, even for thin polyominoes; on the positive side, for hole-free polyominoes that contain no induced $3\times 3$-square, there is a linear-time $4$-approximation.

\subparagraph{Guarding with imprecisions.}
Guarding problems with imprecision have appeared in different contexts.
Dobbins, Kleist, Miltzow, and Rz{\k{a}}{\.z}ewski \cite{areaUniv} pose three interesting open problems concerning the computational complexity of art gallery problems related to imprecision, either concerning the guard positions or the given polygon. 
In \textsc{Robust Guarding}, we are given a hole-free polygon $P$ and a number $k$. 
We place unit disks fully contained
inside the polygon, such that no matter how we choose to place a guard in each disk we get a valid guard set. 
Our setting is closely related.

The notion of \emph{$\alpha$-robust guarding} was recently introduced by Das, Filtser, Katz, and Mitchell~\cite{robustly-guarding-polygons-socg}. 
In particular, for any $\alpha > 0$, a guard $g$ \emph{$\alpha$-robustly guards} a point~$p$ if each segment between $p$ and a point $g'$ of distance at most $\alpha\lVert g-p\rVert$ from $g$ is contained in the polygon.
A crucial distinction is that, in their model, the guard’s permissible movement is governed by its distance from the points it must guard, whereas this is not the case in the concept we consider.
They show that the art gallery problem remains \APX-hard under this notion of visibility. 
Complementarily, they provide $\mathcal{O}(1)$-approximation algorithms, where the approximation factor depends on the robustness parameter $\alpha$. 

Another notion of robustness concerns the number of substantially different directions from which each point can be viewed; 
specifically, a point is visible only when it lies in a triangle spanned by guards~\cite{EfratHM05}.
The authors show that it is \NP-hard to decide whether~$k$ guards suffice to \emph{triangle-guard} a given subset of a polygon; 
moreover, they present a polynomial-time $\mathcal{O}(k\log k)$-approximation when guard positions are restricted to an imposed grid on $P$, where $k$ denotes the size of a smallest guard set on the grid.

A few works on terrain guarding also consider variants with imprecise guards or inputs.
Keikha, Löffler, Saumell, and Valtr~\cite{keikha2025guarding} find a placement of guards on a given subset of a $1.5$D terrain such that the total length of the visible part is maximized. 
McCoy, Zhu, and Dutt~\cite{mccoy2023guarding} study the problem of computing a shortest watchtower on an imprecise terrain, where each vertex lies on a given vertical interval. 
When the location of a watchtower is fixed, the problem in 2D and 3D can be solved with linear programming, which leads to an additive $\varepsilon$-approximation for the general problem. 

\section{Concepts}
\label{sec:preliminaries}
We now introduce the fundamental notions. 
To this end, consider the plane with a partition into unit squares, called \emph{cells}, by cutting with horizontal and vertical lines through the integer grid. 
The \emph{dual graph} of a set of cells has a vertex for each cell and an edge between any two cells sharing a side.
A subregion consisting of a union of cells is a \emph{polyomino} if its dual graph is connected; see, e.g.,~\cref{fig:prelims}. 
Clearly, the vertex coordinates of a polyomino are integers, and vertices are either \emph{convex} (90° interior angle) or \emph{reflex}~(270°). 
We denote edges between two consecutive convex vertices as \emph{\outerdents}.
The cells not contained in a polyomino $P$ are \emph{outer cells} (of $P$). 
A polyomino is \emph{hole-free} if its complement is connected.

\begin{figure}[htb]
	\centering
	\captionsetup[subfigure]{justification=centering}
	\begin{subfigure}{0.5\columnwidth}%
		\centering
		\includegraphics[page=1]{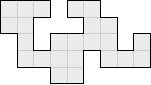}
		\subcaption{}
		\label{fig:prelims-a}
	\end{subfigure}\hfil%
	\begin{subfigure}{0.5\columnwidth}%
		\centering
		\includegraphics[page=2]{prelims}
		\subcaption{}
		\label{fig:prelims-b}
	\end{subfigure}%
	\caption{
		(a) A hole-free polyomino, and 
		(b) an upwards histogram with three (red) top \outerdents and its (hatched) base.}
	\label{fig:prelims}
\end{figure}
We denote the location of a cell by the $x$- and $y$-coordinates of its bottom-left corner.
In a polyomino $P$, a (vertical/horizontal) \emph{strip} is a rectangle of width or height 1, i.e., a $(1\times h)$ or $(w\times 1)$-subpolyomino of $P$. 
A \emph{histogram} is a polyomino that (can be rotated so that it) consists of vertical strips aligned on a horizontal \emph{base strip}, see \cref{fig:prelims}. 
The necessary rotation determines the orientation of the histogram. 
When speaking about a histogram, we usually assume that its base is at the bottom side; 
we also say the histogram is \emph{upwards}. 
A crucial property of a histogram~$H$ is that if a cell $c$ is contained in $H$ so are all cells between~$c$ and the base strip.

\subparagraph{Visibility and neighborhoods.}
We start by defining rectangular visibility also known as r-visibility.
For two points $p$ and $q$ in the plane, let $\rect(p,q)$ denote the smallest axis-aligned rectangle containing them, i.e., their axis-aligned bounding box. 
Specifically, the \emph{visibility region} of a point $x$ in a polyomino $P$ consists of all points $p \in P$ such that $\rect(x,p) \subseteq P$. 

Moreover, we aim to model the situation where a guard does not necessarily stay at its predefined position but may move to any point in a cell at distance at most $k$; its visibility region contains the points that are seen from all possible positions.
Note that the \vregions for all interior points of a cell are equivalent; 
only for boundary points the \vregion might increase. 
However, by the slightest movement from the boundary inside the cell, the enhanced visibility vanishes.
Therefore, without loss of generality we place guards on the midpoints of cells;
these positions can be identified with vertices in the dual graph.

For a \guard{k} $g$ in $P$, we denote the \emph{$k$-neighborhood} of $g$ by $N_k(g)$, i.e., $N_k(g)$ contains all cells in $P$ at distance  at most $k$ from $g$ in the dual graph.
Now, $\viska{0}{g}$ denotes the set of all points $p$ in $P$ such that $\rect(p,g) \subseteq P$ and $\viska{k}{g}:=\bigcap_{c\in N_{k}(g)}\viska{0}{c}$. 
For an illustration, see \cref{fig:intro}. 

\begin{figure}[htb]
	\centering
	\includegraphics[page=2,scale=1.1]{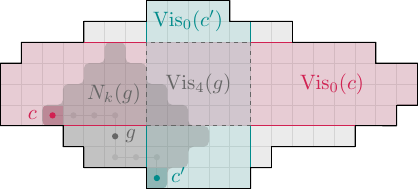}
	\caption{Illustration for \guarding{k} with $k=4$: 
		The neighborhood of a \rguard{k}~$g$ is highlighted in gray.
		Because both cells $c$ and $c'$ are in the $k$-neighborhood of $g$, the $k$-\vregion $\viska{k}{g}$ of $g$ is (a subset of) the intersection of $\viska{0}{c}$ and $\viska{0}{c'}$, which are indicated in red and cyan, respectively.}
	\label{fig:intro}
\end{figure}

Note that $\viska{0}{g}$ is itself a polyomino. 
For simplicity, we say that a \rguard{k} $g$ \emph{guards} or \emph{sees} a point $p$  or subpolyomino $P'$ of $P$, exactly if $p \in \viska{k}{g}$, or $P' \subseteq \viska{k}{g}$, respectively. 
We remark that, for two cells/guards $c$ and $c'$, we usually illustrate $\rect(c,c')$ as the axis-aligned bounding box of the cells rather than their midpoints. 

Finally, we note some basic properties. 
Recall that a set is \emph{orthogonally convex} if the intersection with any axis-aligned line is connected.
\begin{observation}\label{obs:basicproperties}
	The following properties hold for a cell $c$ and a guard $g$ in a polyomino $P$:
	\begin{itemize}
		\item $\viska{k}{g}\subseteq \viska{0}{g}$
		\item symmetry: $c\in \viska{0}{g}\iff g\in \viska{0}{c}$
		\item $\viska{0}{g}$ is connected, orthogonally convex, and hole-free.
	\end{itemize}
\end{observation}
\begin{proof}
	The first statement follows directly from the definition and the fact that $g$ is contained in its $k$-neighborhood. 
	Symmetry follows from the fact that $\rect(g, c) \subseteq P$ in both cases.
	For connectedness, note that every point~$p$ seen by $g$ satisfies $\rect(g, p) \subseteq \viska{0}{g}$, connecting $p$ to $g$.
	To see orthogonal convexity, let two cells $c, c' \in P$ share a coordinate, and let $g$ see both. 
	Then $\rect(g,c)$ and $\rect(g,c')$ are contained in $\viska{0}{g}$, and together contain every cell between $c$ and $c'$. Thus, $\viska{0}{g}$ is orthogonally convex, and hence hole-free.
\end{proof}

\section{Guardability}
\label{sec:first-observations}

We start with the crucial question of whether a polyomino can be \guarded{k}.

\begin{proposition}
	\label{prop:existence}
	While for $k\in\{0,1\}$, every polyomino is \guard{k}{able}, for every $k\geq 2$, there exists a polyomino that is not \guard{k}{able}.
\end{proposition}

\begin{proof}
	For $k\leq 1$, we observe that a \guard{k} always sees the cell that it is placed on. 
	Therefore, by placing one guard in each cell, every polyomino is \guard{k}{able}.
	
	However, for all $k\geq 2$, there exist hole-free polyominoes that are not \guardable{k}. 
	To this end, consider the polyomino consisting of $3$ cells, depicted in \cref{fig:2-restless-guards-with-no-guarding}. 
	As the $k$-neighborhood of every guard $g$ contains $c'$ and $c''$, it follows that $\viska{k}{g}=\{c\}$. \cref{fig:2-restless-guards-with-no-guarding} also illustrates additional examples of non-guardable polyominoes.
\end{proof}
\begin{figure}[htb]
	\centering
	\includegraphics[scale=1.1]{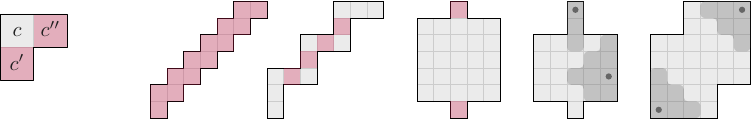}
	\caption{Examples for \guarding{2}: 
		Cells that cannot be guarded are marked in red; 
		in the guardable polyominoes, possible guard sets are indicated.} 
	\label{fig:2-restless-guards-with-no-guarding}
\end{figure}

The fact that not all polyominoes are \guardable{k}, as established in \cref{prop:existence}, naturally raises the question of the computational complexity of testing this property. 
To~address this, we present a polynomial-time algorithm.

\begin{restatable}[{\thmstar{prop:existencePolytime}{proof:existencePolytime}}]{proposition}{existencePolytime}
	\label{prop:existencePolytime}
	There exists an algorithm to decide for a given polyomino~$P$ and an {integer}~$k$ whether $P$ is \guardable{k} in time polynomial in the number of vertices of~$P$.
\end{restatable}

\begin{proof}[Proof sketch]
	Clearly, $P$ is \guardable{k} if and only if all its cells are. Moreover, it is easy to see that
	a cell $c$ is \guardable{k} if and only if $\viska{0}{c}$ contains the \mbox{$k$-neighborhood} of some guard in $P$.
	Hence, two important steps remain. 
	First, as the number of cells might be exponential in the encoded size of $P$, it is crucial to bound the number of interesting cells. 
	Second, we need to efficiently check whether a visibility polyomino contains a \mbox{$k$-neighborhood}; 
	in particular, we do that in a running time independent of $k$.
	
	To this end, we show that there exists a partition $R$ of $P$ into at most $n^2$ rectangles such that any two cells in the same rectangle are either both guardable or not. 
	Specifically, $R$ is obtained by shooting a horizontal and vertical ray at each reflex corner, see~\cref{fig:partition0main}.
		\begin{figure}[htb]
		\centering
			\centering
			\includegraphics[page=2,scale=1.1]{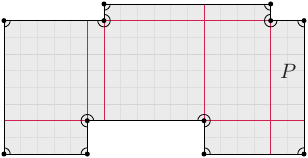}
			\label{fig:partition0-amain}
		\caption{
			The partition $R$ of $P$ is obtained by shooting two (red) rays from each reflex vertex.}
		\label{fig:partition0main}
	\end{figure}
	
	The proposition then follows by showing and combining the following three properties.
	
	\begin{enumerate}[(i)]
		\item \label{item:Pi} For any two cells $c$ and $c'$ in a rectangle in $R$ it holds that: $\viska{0}{c}=\viska{0}{c'}$.
		\item \label{item:Pii} The visibility region $\viska{k}{g}$ of any \rguard{k} $g$ is the union of rectangles from $R$.
		\item \label{item:Piii} There exists a polynomial-time algorithm to check whether the visibility region $\viska{0}{c}$ of a cell $c$ contains the $k$-neighborhood of some guard $g$, and to identify such a guard.\qedhere
	\end{enumerate}

\end{proof}

As the described properties allow us to identify a \rguard{k} for every guardable rectangle in~$R$, \cref{prop:existencePolytime} immediately yields a valid guard set for guardable instances.

\medskip
With techniques similar to the ones used in the proof of \cref{prop:existencePolytime}, one can test whether two cells can be guarded by the same guard. 
This property is crucial for our~algorithms. 

\begin{restatable}[{\thmstar{lem:guardable}{proof:TwoGuardable}}]{lemma}{TwoGuardable}
	\label{lem:guardable}
	For any two given cells $c$ and $c'$ of a polyomino, there exists a polynomial-time algorithm to identify a \rguard{k} that sees both $c$ and $c'$, or decide that no such guard exists.
\end{restatable}

\section{Tight worst-case bounds}

\label{sec:worstCaseBounds}
We present tight worst-case bounds for the cases when \rguard{k} sets exist for all polyominoes. 
We begin with \rguards{0}; this setting is closely related to the problem of guarding polyominoes with general point guards. 
However, \rguards{0} are more restricted, as they must lie within a single cell. 
As a consequence, two \rguards{0} might be required to achieve the visibility of one point guard, as illustrated in \cref{fig:0-guards-point-guards}. 
Therefore, we derive combinatorial bounds that differ from the ones for polyominoes by Biedl et al.~\cite{agpForPolynominoes}. 

\begin{figure}[htb]
	\centering
	\includegraphics[scale=1.1]{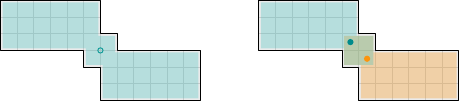}
	\caption{As point guards can be placed on cell borders, it may require two \rguards{0} to achieve the visibility of a single point guard.
	}
	\label{fig:0-guards-point-guards}
\end{figure}

\begin{restatable}[{\thmstar{prop:0-restless-guards-can-guard-polyominoes}{proof:zerorestlessnecessarysufficient}}]{proposition}{zerorestlessnecessarysufficient}
	\label{prop:0-restless-guards-can-guard-polyominoes}
	For a hole-free polyomino $P$ with $n$ vertices and $m\geq 2$ cells, $\floor{\nicefrac{n}{4}}$ as well as $\floor{\nicefrac{m}{2}}$ \rguards{0} are always sufficient and sometimes necessary. 
\end{restatable}

\begin{proof}[Proof sketch]
	\cref{fig:necessary-0} depicts a class of polyominoes establishing the lower bound for both~$n$ and $m$. 
	
	\begin{figure}[htb]
		\centering
		\includegraphics[page=2,scale=1.2]{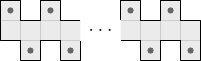}
		\caption{A class of polyominoes requiring $\floor{\nicefrac{n}{4}}$ and $\floor{\nicefrac{m}{2}}$ \rguards{0}.}
		\label{fig:necessary-0}
	\end{figure}

	It is easy to see that each cell on the topmost and bottommost row requires a distinct \rguard{0}. 
	For sufficiency in terms of $n$, we apply a result of Gy\"ori~\cite{orthoDecomposition} to partition the polyomino into at most $\floor{\nicefrac{n}{4}}$ simple shapes that each can be guarded by a single \rguard{0}. 
	Lastly, $\floor{\nicefrac{m}{2}}$ guards suffice because each \rguard{0} sees its own and all neighboring cells. Therefore, we can color the dual graph of $P$ with two colors and choose the smaller set as our \rguard{0} set.
\end{proof}

\medskip

For \rguards{1}, we obtain the following bounds.

\begin{restatable}[{\thmstar{prop:1-restless-guards-can-guard-polyominoes}{proof:onerestlessnecessarysufficient}}]{proposition}{onerestlessnecessarysufficient}
	\label{prop:1-restless-guards-can-guard-polyominoes}
	For a hole-free polyomino $P$ with $n$ vertices and $m$ cells, $\nicefrac{(n-2)}{2}$ as well as $(m-2)$ \rguards{1}, for $m\geq 4$, are always sufficient and sometimes necessary.
\end{restatable}

\begin{proof}[Proof sketch]
\cref{fig:necessary-1} depicts polyominoes used for establishing the lower bounds.

\begin{figure}[htb]
	\centering
	\captionsetup[subfigure]{justification=centering}
	\begin{subfigure}{0.5\columnwidth}%
		\centering
		\includegraphics[page=1]{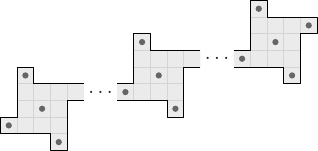}
		\subcaption{}
		\label{fig:necessary-1-a}
	\end{subfigure}%
	\begin{subfigure}{0.3\columnwidth}%
		\centering
		\includegraphics[page=1]{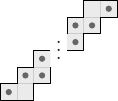}
		\subcaption{}
		\label{fig:necessary-1-b}
	\end{subfigure}%
	\caption{(a) A family of polyominoes that require $\nicefrac{(n-2)}{2}$ \rguards{1} to be entirely guarded.
		(b) Staircase polyominoes require $m-2$ guards for a complete \guarding{1}.}
	\label{fig:necessary-1}
\end{figure}

For the upper bounds, we describe two separate strategies. 
In order to guard $P$ with at most $(m-2)$ \rguards{1}, we identify two distinct cells that can be seen from an adjacent guard position, and place a \rguard{1} on all other cells. 
We prove the existence of two such cells by analyzing configurations in the topmost and bottommost rows of $P$.

The strategy for placing at most $\nicefrac{(n-2)}{2}$ \rguards{1} is more involved. 
The idea is to compute a partition of~$P$ into at most $\nicefrac{(n-2)}{2}$ rectangles, so that a single \rguard{1} in each rectangle suffices.
Interestingly, not every partition into rectangles has this property.
We therefore present a suitable partitioning strategy.
\end{proof}

\section{Computing minimum guard sets is \NP-complete}
\label{sec:complexity}
In this section, we discuss the computational complexity for computing minimum guard sets. In particular, we show that 
the general problem of deciding whether $\ell$ many \rguards{k} are sufficient to guard a given polyomino is \NP-complete.

\begin{restatable}[{\thmstar{prop:np-hardness-for-any-k}{proof:propHardness}}]{theorem}{propHardness}
	\label{prop:np-hardness-for-any-k}
	For each fixed  $k\in \mathbb N$, the decision problem of whether a given polyomino has a \rguard{k} set of size at most $\ell$ is \NP-complete.
\end{restatable}

We leverage insights from Biedl and Mehrabi~\cite{Biedl016} who showed that finding minimum guard sets in orthogonal polygons with holes under rectangular visibility is \NP-hard, even when the polygons are thin. 
They present a simple reduction from \textsc{Minimum Vertex Cover} for plane graphs of maximum degree 3. 
The idea of their reduction is as follows: 
Starting from a rectilinear embedding of the input graph $G$, they construct a polyomino $P$ by ``thickening'' the edges of the embedding. 
It is then straightforward to verify that any guard set for~$P$ corresponds to a vertex cover of $G$, and vice versa. 

We construct similar gadgets but require new ideas, as our guards are both more restricted (since they cannot be placed on the boundary of a cell) and more flexible (because they can move).
Our gadgets are depicted in~\cref{fig:hardness}. 
\begin{figure}[htb]
	\centering
	\captionsetup[subfigure]{justification=centering}
	\begin{subfigure}{0.17\textwidth}%
		\centering
		\includegraphics[scale=0.8,page=1]{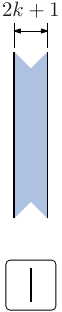}
		\subcaption{}
	\end{subfigure}%
	\begin{subfigure}{0.18\textwidth}%
		\centering
		\includegraphics[scale=0.8,page=2]{hardness_pages}
		\subcaption{}
	\end{subfigure}%
	\begin{subfigure}{0.25\textwidth}%
		\centering
		\includegraphics[scale=0.8,page=4]{hardness_pages}
		\subcaption{}
	\end{subfigure}%
	\begin{subfigure}{0.4\textwidth}%
		\centering
		\includegraphics[scale=0.8,page=6]{hardness_pages}
		\subcaption{}
	\end{subfigure}%
	\caption{
		The gadgets used for reducing from minimum vertex cover. 
		The corridors connect the other gadgets.
		(a)
		A~corridor. 
		(b)
		A~corner gadget.
		(c)
		A~tunnel gadget.
		(d)~A~junction gadget.}
	\label{fig:hardness}
\end{figure}
To incorporate the movement in the neighborhoods, we scale the gadgets so that each corridor is assigned a width or height of $2k+1$. 
Another key difference in our reduction is that
certain gadgets (tunnel and junction gadget) require a `ground set' of guards that must always be present, depicted by the gray guards in~\cref{fig:hardness}. 
This set alone does not suffice to see any neighboring corridor. 
However, placing a single additional guard within the interior of the gadget extends the coverage to all neighboring corridors, as it enables a configuration in which all guards are positioned as indicated by the outlined guards.

\section{Computing small guard sets}
\label{sec:k-restless}

In this section, we present algorithms for guarding polyominoes with \rguards{k}. 
Firstly, we show how to compute a minimum \rguard{k} set for histograms or determine that no such guard set exists. 
Secondly, for $k=0$, a 2-approximation follows by slightly modifying a result from previous work. 
Lastly, we exploit a simple decomposition of hole-free polyominoes into histograms and show that those can be guarded efficiently, yielding a $20k$-approximation algorithm for $k>0$.

\subsection{Optimal restless guard sets for histograms}
In this section, we show how to compute optimal guard sets for histograms.
We begin by establishing a property that will be essential to prove lower bounds. 
For later use, we already state it in a slightly more general form, namely for the more general `embedded histograms';
for now, the reader is welcome to restrict their attention to `real' histograms.

\begin{restatable}[{\thmstar{lem:separation-better}{proof:separation}}]{lemma}{separation}
	\label{lem:separation-better}
	Let $H$ be a histogram (possibly embedded in a hole-free polyomino $P$), $c = (x,y)$ and $c' = (x',y')$ be two cells in $H$ with $y \geq y'$, and $S$ be the horizontal set of cells from $(x',y')$ to $(x,y')$.
	If a \rguard{k} $g$ (in $P$) sees both $c$ and $c'$, then $S\subseteq H$. 
	
	In particular, if $S\nsubseteq H$, then no \rguard{k} can see both $c$ and $c'$. 
\end{restatable}

We highlight that the proof of the general statement crucially relies on the fact that~$P$ has no holes.
The proof for histograms is much simpler, with the underlying argument illustrated in~\cref{fig:separation}.
In particular, if~a guard~$g$ sees both cells $c$ and $c'$, then the rectangles $\rect(g,c)$ and $\rect(g,c')$ are contained in~$H$. 
Moreover, as $H$ is a histogram, the extensions of the rectangles to the base are also contained in~$H$. Hence, $S\subseteq H$.

\begin{figure}[htb]
	\centering
		\captionsetup[subfigure]{justification=centering}
	\begin{subfigure}{0.5\textwidth}
		\centering
		\includegraphics[page=1]{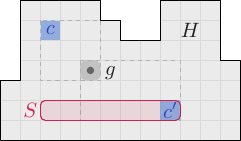}
		\subcaption{}
		\label{fig:separation-a}
	\end{subfigure}%
	\begin{subfigure}{0.5\textwidth}
		\centering
		\includegraphics[page=2]{separation_1.2}
		\subcaption{}
		\label{fig:separation-b}
	\end{subfigure}%
	\caption{Illustration of \cref{lem:separation-better} for histograms.
		(a) If $g$ sees both cells $c$ and $c'$, then the cells in $S$ are contained in the histogram $H$. 
		(b) Because $S$ contains outer cells, no \rguard{k} is able to see both cells $c$ and $c'$.}
	\label{fig:separation}
\end{figure}

\medskip

With \cref{lem:separation-better} at hand, we are ready to present an algorithm for computing an optimal set of \rguards{k} for any histogram, or decide that no such set exists.

\begin{restatable}[{\thmstar{thm:k-guard-opt-hist}{proof:theoremOptHist}}]{theorem}{theoremOptHist}
    \label{thm:k-guard-opt-hist}
    There exists an algorithm that, given any histogram $H$ with $n$ vertices and any $k\in \mathbb{N}$, computes an optimal \rguard{k} set in time polynominal in the number of vertices of $H$, or decides that $H$ is not \guardable{k}. 
\end{restatable}

\begin{proof}[Proof sketch]
We assume that $H$ is upwards.
The key to optimally \guarding{k} a histogram $H$ lies in partitioning $H$ into \emph{blocks}, which are maximal horizontal rectangles that touch the histogram boundary to both sides, as illustrated in~\cref{fig:blocks}. 
We observe that the dual graph of the blocks forms a rooted tree~$T$. 
For each block, we assign a \emph{key cell}, which is an arbitrary cell incident to the top boundary of $H$ (the corresponding adjacent outer cells are marked with a cross in \cref{fig:blocks}).
We establish the following properties: 

\begin{enumerate}[(i)]
	\item If $H$ is \guardable{k}, then each block can be guarded by a single \rguard{k}. 
	\item A~\guard{k} sees a block if and only if it sees the block's key cell.
	\item Using \cref{lem:separation-better}, we show that the blocks seen by a \rguard{k} induce a directed path in $T$, and all such paths can be identified efficiently using \cref{lem:guardable}. 
\end{enumerate}

It follows that \guarding{k}~$H$ reduces to covering $T$ with a minimum number of directed paths. 
This can be done optimally by applying a greedy method starting at the leaves in $T$. 
To prove minimality, we select one key cell per selected path in $T$ and show that these cells form a witness set, meaning that no two of the selected key cells can be seen by the same \rguard{k}.
\end{proof}

\begin{figure}[htb]
    \centering
    \includegraphics[scale=1.1]{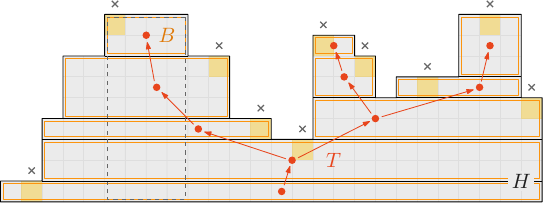}
    \caption{
        Optimally \guarding{k} a histogram $H$ is based on a partition of $H$ into rectangular blocks (highlighted in orange) that form a rooted tree $T$. 
        The extension of a block reaches down to the histogram's base, as outlined for block $B$ in dashed lines.
        Each block is assigned a key cell within the respective block, visualized in yellow.}
    \label{fig:blocks}
\end{figure}

\subsection{Approximation for \guarding{0}}

Before we focus on the general case of \guarding{k} polyominoes, we briefly remark that previous work implies a \mbox{$2$-approximation} algorithm for \guarding{0} polyominoes without holes. 
Worman and Keil~\cite{worman07} devise a polynomial-time algorithm that computes a minimum guard set for hole-free polyominoes with classic rectangular visibility. 
To leverage this result, the required insight is that two \rguards{0} can replace a point guard.

\begin{restatable}[{\thmstar{lem:twoRestlessOnePoint}{proof:twoRestlessOnePoint}}]{lemma}{twoRestlessOnePoint}\label{lem:twoRestlessOnePoint}
Let $p$ be a point guard in a polyomino $P$. 
	We can efficiently identify two \rguards{0} $g$ and $g'$ such that $\viska{0}{g} \cup \viska{0}{g'}$ covers the visibility region of $p$.
\end{restatable}

The high-level idea is visualized in~\cref{fig:two-approx}.
In particular, observe that we can partition the visibility region of any point guard $p$ into two histograms.
Then, by carefully placing two \rguards{0}, those guards together cover the corresponding region.

\begin{figure}[htb]
	\centering
	\includegraphics[page=2]{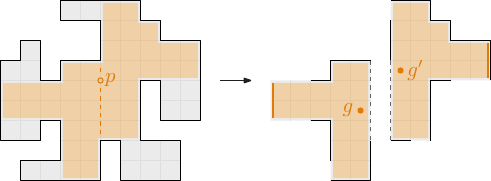}
	\caption{The visibility region of a point guard can be split into two histograms that can be guarded with a single \rguard{0} each.
	}
	\label{fig:two-approx}
\end{figure}

\cref{lem:twoRestlessOnePoint} implies that a set of point guards for a polyomino can be transformed into a set of \rguards{0} whose cardinality is at most twice that of the original.
Therefore, using the algorithm of Worman and Keil~\cite{worman07} as a subroutine, we obtain a $2$-approximation.

\begin{corollary}
	\label{prop:0-approx-2}
	There exists a $2$-approximation
	for \guarding{0} hole-free polyominoes. 
\end{corollary}

\subsection{Approximation algorithm for hole-free polyominoes}

We now turn our attention to general hole-free polyominoes and present a $20k$-approximation.

\begin{restatable}{theorem}{thmConstantFactorHoleFreePolyominoes}
    \label{thm:k-guard-poly}
    For $k>0$, there exists a $20k$-approximation algorithm that, given a  \guardable{k} hole-free polyomino $P$, computes a \rguard{k} set in time polynomial in the number of vertices of~$P$.
\end{restatable}

Our proof consists of three steps. 
First, we decompose the polyomino into histograms. 
Second, we guard each histogram. 
Finally, we argue that the union of these guard sets yields a good guard set for the entire polyomino. 
The first and last step follow from previous work.
The interesting part lies in the second step.
A major challenge is that the boundaries of the histograms are not necessarily boundaries of the polyomino, allowing the \rguards{k} to leave the histogram. 
We now explain the three steps.

\subparagraph*{Decomposing a hole-free polyomino into histograms.}
We apply a simple method by Levcopoulos~\cite{levcopoulos} for decomposing a hole-free polyomino $P$ into a set of histograms: 
In each iteration, we construct a new maximal histogram starting at a given base. 
For an illustration, consider~\cref{fig:5}.
The first histogram is obtained by choosing an outer dent $e$; note that, without loss of generality, we choose a horizontal outer dent in~\cref{fig:5}.
Then, we construct a maximal histogram with base $e$; 
you may imagine sweeping over $P$ with $e$. 
Let $\mathcal{H}$ be the union of the already constructed histograms. 
We iteratively construct the next histogram from a maximal segment leading from $P \setminus \mathcal{H}$ to $\mathcal{H}$, until $\mathcal{H}=P$.

\begin{figure}[htb]
	\centering
	\includegraphics[]{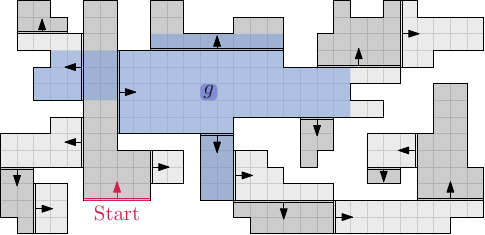}
	\caption{
		Partitioning a polyomino into histograms, with red being the starting edge. 
	}
	\label{fig:5}
\end{figure}

The histograms constructed by the aforementioned approach fulfill the following properties, which we also visualize in~\cref{fig:embedded-intro}. 
Note that for simplicity of the arguments and to match the figures, we assume that the histogram $H$ is upwards.
\begin{enumerate}[(i)]
	\item By maximality, a cell adjacent to the top boundary of $H$ but not contained in $H$ is an outer cell of $P$.
	\item All cells adjacent to the unique bottom edge of $H$ but not contained in $H$ are either all contained in~$P$ or all outer cells of $P$. 
	We refer to $H$ having an \emph{\fullbase} or a \emph{\emptybase}, respectively. 
\end{enumerate}
We~call a histogram-shaped subpolyomino~$H \subseteq P$ satisfying these properties an \emph{embedded histogram}.
We conclude with \cref{lem:partition}.

\begin{figure}[htb]
	\centering
	\includegraphics[page=2]{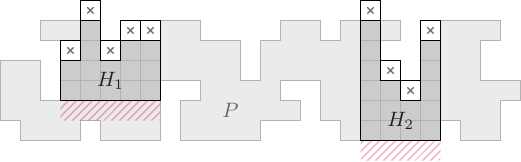}
	\caption{Two embedded histograms within a polyomino $P$: $H_1$ has an \fullbase, while $H_2$ has a \emptybase, both hatched in red.
	Outer cells adjacent to the top boundary of $H_1$ and $H_2$ are marked by crosses.}
	\label{fig:embedded-intro}
\end{figure}

\begin{lemma}[Levcopoulos~\cite{levcopoulos}, Lemma 4]
	\label{lem:partition}
	A hole-free polyomino with $n$ vertices can be partitioned into embedded histograms in $\bO(n)$ time.
\end{lemma}

\subparagraph*{Challenges when handling embedded histograms.}
As mentioned above, the transition from histograms to embedded histograms might seem like a minor modification, but has fundamental effects on potential guarding strategies.
We highlight some of the encountered difficulties in~\cref{fig:embedded_difficulties}:
For guarding histograms, we used a natural partition into rectangular blocks, each of which could be covered by a single guard.
As visualized in~\cref{fig:embedded_difficulties-a}, a block within an embedded histogram in turn might require arbitrarily many \rguards{k}. 
Moreover, the example depicted in~\cref{fig:embedded_difficulties-b} shows that optimal solutions may require \vregions that do not touch the histogram boundary. 
We therefore deal with an inherently two-dimensional covering problem that cannot be simply reduced to one-dimensional strategies. 
Furthermore, even if we are able to compute optimal solutions for blocks, these may not remain optimal when merged, as illustrated in~\cref{fig:embedded_difficulties-c}. 

\begin{figure}[htb]
	\captionsetup[subfigure]{justification=centering}
	\begin{subfigure}{0.25\textwidth}%
		\centering
		\includegraphics[page=1]{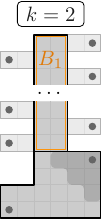}
		\subcaption{}
		\label{fig:embedded_difficulties-a}
	\end{subfigure}%
	\begin{subfigure}{0.45\textwidth}%
		\centering
		\includegraphics[page=2]{embedded_difficulties_pages}
		\subcaption{}
		\label{fig:embedded_difficulties-b}
	\end{subfigure}%
	\begin{subfigure}{0.3\textwidth}%
		\centering
		\includegraphics[page=3]{embedded_difficulties_pages}
		\subcaption{}
		\label{fig:embedded_difficulties-c}
	\end{subfigure}%
	\caption{Illustration of the challenges for guarding embedded histograms.
		(a) Guarding the block $B_1$ with \rguards{k} may require arbitrarily many guards.
		(b) \Vregions in optimal solutions may not touch the histogram boundary.
		(c) Combining optimal solutions for blocks may be suboptimal; $B_i$ needs $i$ \rguards{k}, but four \rguards{k} suffice for $B_2 \cup B_3$.}
	\label{fig:embedded_difficulties}
\end{figure}

\subparagraph*{Guarding embedded histograms.}
Before we present our algorithm, we start by guarding strips inside a polyomino. 
Recall that strips are rectangles of width or height 1.

\begin{lemma}
	\label{lem:strip-guarding-interval-cover}
	Let $P$ be a polyomino and let $S$ be a strip in $P$.
	We can compute a minimum set of \rguards{k} in $P$ covering $S$ in time polynomial in the number of vertices of~$P$, or decide that no such set exists.
\end{lemma}
\begin{proof}
	We show that this problem reduces to covering an interval with given subintervals.
	If~a~guard~$g$ sees two cells $s$ and $s'$ in the strip $S$, then $g$ also sees all cells between~$s$ and~$s'$ because $\rect(s,c)$ and $\rect(s',c)$ cover $S$ for all $c\in N_k(g)$. 
	Therefore, the visibility region of a guard restricted to $S$ forms a substrip of $S$. 
	Consequently, covering a strip with substrips corresponds to covering an interval with subintervals, see also \cref{fig:interval-cover}. 
	By \cref{{lem:guardable}}, we can check for any two cells in $S$ if there exists a \rguard{k} seeing both cells in polynomial time. 
	However, we need to restrict our attention to a polynomial number of interesting pairs of cells.
	To this end, note that $S$ intersects at most $n$ rectangles of $R$, so it suffices to test a quadratic number of cell pairs, and thus subintervals.
	(We may even further restrict our attention to maximal subintervals, therefore also $\bO(n\log n)$ pairs suffice when using binary search.)
	
	\begin{figure}[htb]
		\centering
		\includegraphics[]{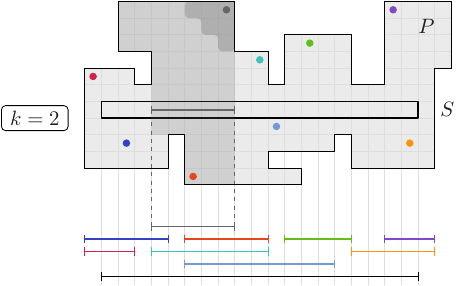}
		\caption{Illustrating the proof idea of \cref{lem:strip-guarding-interval-cover}: 
		Guarding a strip $S$ in $P$ corresponds to computing a minimum set of intervals to cover a given interval.}
		\label{fig:interval-cover}
	\end{figure}

	Finally, the interval cover problem can be solved by a greedy algorithm that iteratively chooses an interval that covers the first uncovered point and maximizes its end. 
	It is easy to see that this greedy strategy is optimal; 
	any optimal interval cover can be transformed into the greedy solution by repeatedly exchanging the first interval that differs between both solutions with an interval that covers possibly more points to the end.
	We remark that this strategy also identifies when no \rguard{k} set exists.
\end{proof}

\medskip

Using \cref{lem:strip-guarding-interval-cover,lem:separation-better}, we are now ready to describe an approximation algorithm. 

\begin{lemma}
    \label{lem:k-guard-embedded}
    There exists a $4k$-approximation algorithm to compute, given a hole-free \guardable{k} polyomino $P$ and an embedded histogram $H$ within $P$, a set of \rguards{k} contained in $P$ that sees $H$. 
	Its running time is polynomial in the number $n$ of vertices of $P$. 
\end{lemma}

\begin{proof}
    Let $P$ be a hole-free polyomino, and let $H$ be a histogram embedded in~$P$.
    We start by partitioning $H$ into three subsets, a bottom set $B$, a top set $T$, and a rest $R$. For an illustration, see~\cref{fig:k-guard-emb}. 
    In a polyomino, we call the set of all cells that share the same $y$-coordinate a \emph{row}. 
	We count rows from bottom to top, i.e., the first row of a histogram is its baseline.
    The bottom set $B$ contains the $2k-1$ bottommost rows of $H$. 
    We then partition the remaining cells into maximal horizontal strips. 
	All strips of width $\leq 2k$ are contained in the top set $T$ and the rest is assigned to the set $R$. 
    
    \begin{figure}[htb]
		\vspace{1mm}
    	\centering
    	\includegraphics[]{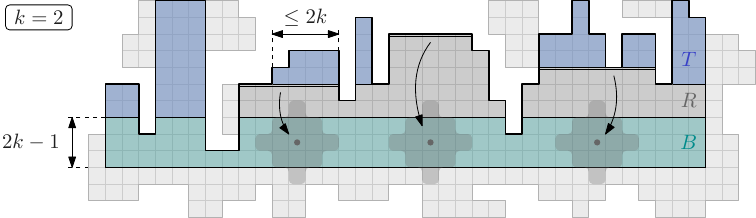}
    	\caption{An embedded histogram and its bottom set $B$, top set $T$ and rest set $R$.
		To guard $R$, we place a guard (gray) at the $k$th row centered below all top \outerdents in~$R$ (double lines).}
    	\label{fig:k-guard-emb}
		\vspace{1mm}
    \end{figure}

	\pagebreak
    
    We now describe how to \guard{k} all three sets separately. Let $\opt$ denote the number of \rguards{k} in an optimal solution.
    In particular, we use at most
    $(2k-1)\cdot\opt$ guards for $B$, $2k\cdot\opt$ guards for $T$, and $1\cdot \opt$ guards for $R$. 
	Together, this yields $4k\cdot\opt$ \rguards{k} and implies the claim.

    Applying \cref{lem:strip-guarding-interval-cover}, we compute an optimal guard set for every maximal horizontal strip in~$B$.
    If several horizontal strips lie in the same row, \cref{lem:separation-better} guarantees that the union of optimal guard sets remains optimal for the union of the strips in one row. 
	Note that \cref{lem:separation-better} crucially requires~$P$ to be hole-free.
    Clearly, optimally guarding any subset of~$P$ (i.e., any strip) requires at most $\opt$ guards.
    Therefore, taking the union of the guard sets for the~$2k-1$ rows in $B$ yields a set containing at most $(2k-1)\cdot\opt$ \rguards{k}.

    For guarding $T$, we apply \cref{lem:strip-guarding-interval-cover} to \emph{vertical} strips. 
    We observe that in $T$, no \rguard{k} can see cells that are more than $2k$ cells apart; 
	otherwise, by \cref{lem:separation-better}, $T$ would contain a horizontal strip of width $>2k$.
    We can therefore distribute our vertical strips into~$2k$ sets $T_1,\dots, T_{2k}$ based on their $x$-coordinate $\text{mod }2k$.
    In every set $T_i$, \cref{lem:separation-better} ensures that no two vertical strips can be guarded by the same \rguards{k}. 
	Therefore, each $T_i$ is optimally guarded, and the total number of guards for $T$ is at most $2k\cdot\opt$ \rguards{k}. 
    
    We consider $R$ as a set of histograms. 
	By construction, every top \outerdent of histograms in $R$ has a length exceeding $2k$. 
    To guard $R$, we place a \rguard{k} at the $k$th row in $B$ at the center of each top \outerdent in $R$, as illustrated in \cref{fig:k-guard-emb}. 
   	We argue that this set of guards is valid.
   	For a fixed outer dent, consider the rectangle $r$ that contains its top edge and reaches to the histogram base; 
	if the histogram has an \fullbase, we extend $r$ vertically to also include these cells. 
	Hence $r$ has width $>2k$.
	Its height is $>2k$ (and a boundary to the bottom), or~$>2k+1$, if the histogram has an \fullbase.
	In both cases, the $k$-neighborhood of our \rguard{k} $g$ is contained in~$r$. 
	Therefore,~$g$ guards~$r$, and every cell below an outer dent is guarded. 
	By the structure of histograms, this implies that every cell in $R$ is seen.
   	Finally, by \cref{lem:separation-better}, each \outerdent requires a distinct \rguard{k}. 
	Therefore, the constructed guard set for $R$ is minimum. 
	
	Lastly, we modify the described strategy to obtain a running time independent of $k$. 
	In particular, we apply \cref{lem:strip-guarding-interval-cover} only to a subset of at most $n^2$ \emph{crucial} vertical and horizontal strips.
	Let $R$ be the partition of $P$ into rectangles obtained by shooting a horizontal and vertical ray at each reflex corner, see also~\cref{fig:partition0main}.
	By \cref{claim:partition}(\ref{item:ii}), the visibility region of a guard is a union of rectangles in $R$. 
	Therefore, if two strips intersect the same set of rectangles, a \rguard{k} set for one strip also sees the other strip. Hence, it suffices to compute a guard set for one of these strips. 
	We can easily (and independent of $k$) identify a set of \emph{crucial} strips each of which we can charge to a rectangle of $R$. 
	Each rectangle is charged by at most one vertical and one horizontal strip.

	This completes the proof.
\end{proof}

\subparagraph*{From embedded histograms to polyominoes.}
It remains to show that an approximation algorithm for embedded histograms implies an approximation for polyominoes. 
We show that the approximation factor increases by a factor of $5$, which follows by an adaptation of a result of Lingas, Wasylewicz, and \.Zyli\'nski~\cite{3approx}.
In particular, they argue that in a partition into histograms, any rectangular visibility polygon of any point guard is contained in the union of at most five histograms. 
For an example, consider the (blue) visibility region \viska{0}{g} of $g$ in~\cref{fig:5}.
As $\viska{k}{g}\subseteq \viska{0}{g}$, this property carries over to \rguards{k}.

\begin{restatable}[{\thmstar{lem:5}{proof:fiveapproxlemma}}]{lemma}{fiveapproxlemma}
	\label{lem:5}
	An $\alpha$-approximation for computing \guard{k} sets for embedded histograms yields a $5\alpha$-approximation for \guarding{k} hole-free polyominoes. 
\end{restatable}

\subparagraph*{The final step.}
Together, \cref{lem:5,lem:k-guard-embedded,prop:existencePolytime} yield the claim. First, we check with \cref{prop:existencePolytime} whether $P$ is \guardable{k}. If so, we partition $P$ into histograms and use \cref{lem:k-guard-embedded} to compute a $4k$-approximation for each embedded histogram. By \cref{lem:5}, the union of these guard sets yields a $20k$-approximation for $P$. 
All steps have runtimes that only depend on the number $n$ of vertices of $P$ and $k$.

\section{Improved approximation factor for small $k$}
\label{sec:better-approx-fixed-k}

With~\cref{thm:k-guard-poly} we established an approximation algorithm for arbitrary $k$. 
For small $k$, however, a refined analysis of the arguments in the proof of~\cref{thm:k-guard-poly} allows us to obtain stronger guarantees. 
In particular, for $k=1$, we can optimally guard each of the three sets of the decomposition  (separately).
This yields the following result:

\begin{restatable}[{\thmstar{thm:1-guard-poly}{proof:thmBetterAlgorithmForHoleFreeKOne}}]{theorem}{thmBetterAlgorithmForHoleFreeKOne}
	\label{thm:1-guard-poly}
	For any constant $\varepsilon > 0$, there exists a $(10+\varepsilon)$-approximation algorithm that, given a hole-free polyomino $P$, computes a \rguard{1} set  in time polynomial in the number of vertices of $P$.
\end{restatable}

\section{Conclusion}
\label{sec:conclusions}
Motivated by locational uncertainty, we introduced the art gallery problem with restless guards that may move in the vicinity of their prescribed location. 
Full coverage must be guaranteed no matter the exact position of guards.
We investigated this problem for polyominoes with rectangular visibility and establish tight worst-case bounds, show \NP-hardness, and develop approximation algorithms.

    \bibliography{references}

    \appendix
    \restatingthmtrue
    

\section{Omitted details for~\texorpdfstring{\cref{sec:first-observations}}{Section~\ref{sec:first-observations}}}
\label{app:first-observations}

In this section, we lay the groundwork for efficiently answering visibility queries inside a polyomino.
In particular, we prove \cref{prop:existencePolytime,lem:guardable}.

\hypertarget{proof:existencePolytime}{}
\existencePolytime*

\begin{proof}
	Clearly, $P$ is \guardable{k} if and only if all its cells are.
	Therefore, the proof is based on the following characterization.
	
	\begin{claim}
		\label{claim:char}
		A cell $c$ in $P$ is \guardable{k} if and only if $\viska{0}{c}$ contains the $k$-neighborhood of some guard position in $P$. 
	\end{claim}
	
	\begin{claimproof}
		The cell $c$ is \guardable{k} if and only if there exists a guard $g$ with $c\in \viska{k}{g}$ and thus $c\in \viska{0}{g'}$ for all $g'\in N_k(g)$. 
        By symmetry of rectangular visibility, this holds if and only if $g'\in \viska{0}{c}$ for all $g'\in N_k(g)$, i.e., $\viska{0}{c}$ contains the $k$-neighborhood of~$g$.\claimqedhere
	\end{claimproof}
	
	There thus remain two important steps. 
    First, as the number of cells might be exponential in the encoding of $P$, it is crucial to bound the number of interesting cells. 
    Second, we need to efficiently check whether the \vregion contains a $k$-neighborhood; 
    in particular, we want to do so in a running time independent of $k$.
	
	To this end, we show that there exists a partition of $P$ into at most $n^2$ rectangles such that any two cells in the same rectangle are either both guardable or not. 
    Moreover, all \vregions can be expressed as a union of rectangles from the partition.
	In particular, we partition $P$ by shooting two rays, one horizontal and one vertical, from every reflex vertex. 
    This yields a partition $R$ into at most $n^2$ rectangles. 
    For an illustration, consider \cref{fig:partition0}.
	
	\begin{figure}[htb]
		\centering
		\captionsetup[subfigure]{justification=centering}
		\begin{subfigure}{0.5\columnwidth}%
			\centering
			\includegraphics[page=2]{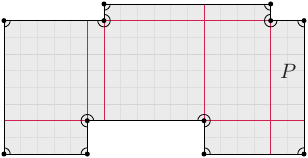}
			\subcaption{}
			\label{fig:partition0-a}
		\end{subfigure}%
		\begin{subfigure}{0.5\columnwidth}%
			\centering
			\includegraphics[page=4]{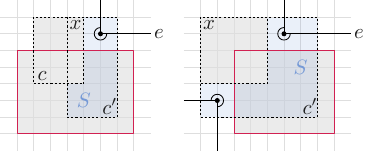}
			\subcaption{}
			\label{fig:partition0-b}
		\end{subfigure}%
		\caption{Illustration for the proof of \cref{prop:existencePolytime}.
			(a) The partition $R$ of $P$ is obtained by shooting two (red) rays from each reflex vertex. 
			(b) A cell $x\in\viska{0}{c}\setminus\viska{0}{c'}$ implies the existence of a vertex in $\viska{0}{c'}$, yielding a ray separating $c$ from $c'$.} 
		\label{fig:partition0}
	\end{figure}
	
	We show that $R$ has the desired properties.
	
	\begin{claim}\label{claim:partition}
		The partition $R$ of $P$ into at most $n^2$ rectangles has the following properties:
		\begin{enumerate}[(a)]
			\item \label{item:i} For any two cells $c$ and $c'$ of a rectangle $r$ in $R$ it holds that: $\viska{0}{c}=\viska{0}{c'}$.
			\item \label{item:ii} The visibility region $\viska{k}{g}$ of any \rguard{k} $g$ is the union of rectangles from $R$.
		\end{enumerate} 
	\end{claim}
	\begin{claimproof}
		\begin{alphaenumerate}
			\item
			Suppose there exists a cell $x\in\viska{0}{c}\setminus\viska{0}{c'}$. 
			Then the set $S:=\rect(c',x)\setminus \rect(c,x)$ is not entirely contained in $P$; 
			otherwise $x\in \viska{0}{c'}$. 
			This implies that some edge $e$ intersects $S$, see \cref{fig:partition0-b} for an illustration.
			Note that $S$ is bounded by a vertical edge from $\viska{0}{c}$ and a horizontal edge from $r$ or vice versa.
			Therefore, the edge $e$ must end in $S$; 
			otherwise its interior intersects $\rect(c,x)$ or $r$, a contradiction to the fact that $c$ sees $x$ or that $r$ contains no vertex in its interior.
            In other words, if $e$ ends in $S$, then $S$ contains a vertex. 
			This implies a ray separating $c$ and $c'$; this is a contradiction to the fact that $c$ and $c'$ stem from the same rectangle. 

			\item By definition, the visibility region $\viska{k}{g}$ of a \guard{k} $g$ is the intersection of $\viska{0}{c}$ for all $c\in N_{k}(g)$, i.e., $\viska{k}{g}=\bigcap_{c\in N_{k}(g)}\viska{0}{c}$. 
            Therefore, it suffices to show that $\viska{0}{c}$ is a union of rectangles of $R$.
			By~\cref{obs:basicproperties} we have the symmetry  
			$a\in \viska{0}{c}\iff c\in \viska{0}{a}$. 
			Therefore, if $a\in \viska{0}{c}$, and $a'$ lies in the same rectangle as $a$, we have $c\in\viska{0}{a}=\viska{0}{a'}$ by \hyperref[item:i]{{\sffamily{\textbf{(a)}}}} and hence $a'\in \viska{0}{c}$. 
			In other words, if $c$ sees one cell of a rectangle in $R$ it sees the entire rectangle. \claimqedhere
		\end{alphaenumerate}
	\end{claimproof}
	
	With property \hyperref[item:i]{{\sffamily{\textbf{(a)}}}} of \cref{claim:partition} at hand it remains to pick an arbitrary cell in each rectangle in $R$ and verify that it is \guardable{k}.
	To this end, we introduce \cref{clm:checkNeighborhood}.

	\begin{claim}\label{clm:checkNeighborhood}
		There exists a polynomial-time algorithm to check whether the visibility region $\viska{0}{c}$ of a cell $c$ contains the $k$-neighborhood of some guard position in the polyomino $P$.
	\end{claim}
	\begin{claimproof} 
		By \cref{claim:partition}(\ref{item:ii}), $\viska{0}{c}$ is the union of at most $n^2$ rectangles from $R$.
		For each rectangle~$r$ of $R$ in $\viska{0}{c}$, we check if a \guard{k} can be placed in $r$, i.e., if the neighborhood of any cell in $r$ is contained in $\viska{0}{c}$. 
		Let $C$ denote the candidate region for the guard placements in $r$; 
		initially, $C:=r$. 
		We classify the boundary edges of $\viska{0}{c}$ into two sets, the ones that are subsets of boundary edges of $P$, and the remaining ones, which we call the set of \emph{exits} $X$. 
		The name reflects the property that a guard may potentially leave $\viska{0}{c}$ by traversing such an edge, see also \cref{fig:exists-a}. 
		
		\begin{figure}[htb]
			\centering
			\captionsetup[subfigure]{justification=centering}
			\begin{subfigure}{0.6\textwidth}%
				\centering
				\includegraphics[page=1]{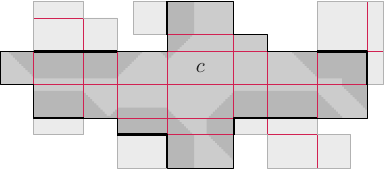}
				\subcaption{}
				\label{fig:exists-a}
			\end{subfigure}\hfil
			\begin{subfigure}{0.3\textwidth}%
				\centering
				\includegraphics[page=2]{exits.pdf}
				\subcaption{}
				\label{fig:exists-b}
			\end{subfigure}\hfil
			\caption{Illustration for the proof of \cref{clm:checkNeighborhood}. 
				(a) Example of $\viska{0}{c}$ of a cell $c$.
				Exits are shown as bold lines and cells in $\viska{0}{c}$ that may reach an exit in at most $k$ steps are shaded gray. 
				(b) For a rectangle $r$ in $R$ and a single exit, the induced cut of cells at most $k$ positions away from the exit is either horizontal, vertical or diagonal.}
			\label{exists}
		\end{figure}
		
		We now iterate over all exits in $X$. 
		Our goal is to remove all cells from $C$ that are too close to the respective exit, i.e., if the exit can be traversed in at most $k$ steps. 
		If $C$ contains a cell $g$ after all exits have been considered, the $k$-neighborhood of $g$ is contained in $\viska{0}{c}$ and we thus found a suitable \rguard{k}; otherwise $c$ is not guardable from a guard in $r$.
	
		It remains to explain how to do these computations in		polynomial time.
		To this end, note that $\viska{0}{c}$ is orthogonally convex and connected, and hence, a shortest path within $\viska{0}{c}$ matches the Manhattan distance. 
		Secondly, we observe that for an exit $x\in X$, which can be obtained by a horizontal or vertical translation of a side of $r$, a horizontal or vertical cut can remove all cells of $C$ which are too close to $x$. 
		For all other exits, a unique corner cell of $r$ is closest. 
		In this case, all close cells can be removed by cutting $C$ with a diagonal line. 
		These cuts can be computed by considering the Manhattan distances of the corners of $r$ and endpoints of each exit, 
		see also \cref{fig:exists-b}. 
		In the end, it suffices to consider the eight most extreme cuts in every direction to test whether $C$ contains a cell of $P$.
	\end{claimproof}

	This concludes the proof. 
\end{proof}

The techniques used to prove \cref{clm:checkNeighborhood} can now easily be extended to test whether two cells are seen by the same guard, a crucial tool for all our algorithms that work for arbitrary~$k$. 

\hypertarget{proof:TwoGuardable}{}
\TwoGuardable*

\begin{proof}
	If a \guard{k} sees two cells $c$ and $c'$, then its $k$-neighborhood is contained in both $\viska{0}{c}$ and $ \viska{0}{c'}$.
	Therefore, we aim to check whether  $V:=\viska{0}{c}\cap \viska{0}{c'}$ contains the $k$-neighborhood of some guard.
	As before, $V$ is a union of at most $n^2$ rectangles from $R$ by \cref{{claim:partition}}(\ref{item:ii}).
	As an intersection of orthogonal convex sets, each connected component of $V$ is orthogonally convex. 
	Thus we apply the strategy of \cref{clm:checkNeighborhood} for each connected component instead of  $\viska{0}{c}$.
\end{proof}

\section{Omitted details for~\texorpdfstring{\cref{sec:worstCaseBounds}}{Section~\ref{sec:worstCaseBounds}}}

\label{app:worstCaseBounds}

In the following, we provide the missing details of our results on the number of necessary and sufficient \rguards{k} required to guard polyominoes. 
We present these bounds in terms of both the number of vertices $n$ and the number of cells $m$ of a hole-free polyomino for \rguards{k} with $k=0$ and $k=1$. 

\head{\subsection{\rguards{0}}}

\hypertarget{proof:zerorestlessnecessarysufficient}{}
\zerorestlessnecessarysufficient*

\begin{proof}
We split the proof into three claims.

\begin{claim}
	\label{lem:nec-0}
	For all $n$ and $m$, there exist polyominoes for which $\floor{\nicefrac{n}{4}}$ and $\floor{\nicefrac{m}{2}}$ \rguards{0} are necessary.
\end{claim}
\begin{claimproof}
	For an illustration, consider~\cref{fig:necessary-0A}. 
	Let $t$ be a fixed integer and consider the polyomino consisting of $m=2t$ cells and $n=4t$ vertices. 
    The cells $b_1,\dots,b_t$ form a horizontal strip and a cell $a_i$ is above $b_i$ if $i$ is odd and below if $i$ is even. 
    Note that $a_i$ can only be guarded by $a_i$ or $b_i$. 
    Consequently, at least $t$ \rguards{0} are necessary. 
    Moreover, attaching an additional cell yields polyominoes with odd $m$, and the $n=4t+2$ vertices with the same bounds on the number of necessary guards. 
    This implies the claim.
\end{claimproof}

\begin{figure}[htb]
	\centering
	\captionsetup[subfigure]{justification=centering}
	\begin{subfigure}{0.33\columnwidth}%
		\centering
		\includegraphics[page=8]{necessary_0}
		\subcaption{}
		\label{fig:necessary-0A}
	\end{subfigure}%
	\begin{subfigure}{0.33\columnwidth}%
		\centering
		\includegraphics[page=1]{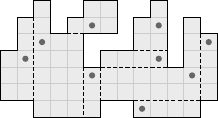}
		\subcaption{}
		\label{fig:suff-0-a}
	\end{subfigure}%
	\begin{subfigure}{0.33\columnwidth}%
		\centering
		\includegraphics[page=5]{sufficient-0}
		\subcaption{}
		\label{fig:suff-0-b}
	\end{subfigure}%
	\caption{
		\label{fig:suff-0}
		(a) A family of polyominoes requiring $\floor{\nicefrac{n}{4}}$ and $\floor{\nicefrac{m}{2}}$ \rguards{0}.
		(b)~Decomposing a polyomino into at most $\floor{\nicefrac{n}{4}}$ orthogonal pieces with $\leq 6$ vertices and guarding each with a single \rguard{0}. 
		(c) Coloring the dual graph of a polyomino black and white to find a guard set with at most $\floor{\nicefrac{m}{2}}$ \rguards{0}.
	}
\end{figure}

\begin{claim}
    \label{lem:suff-nec-0-n}
    $\floor{\nicefrac{n}{4}}$ \rguards{0} are always sufficient.
\end{claim}

\begin{claimproof}
    As shown by Gy\"ori~\cite{orthoDecomposition}, every hole-free orthogonal polygon can be decomposed into at most $\floor{\nicefrac{n}{4}}$ orthogonal pieces, each having 4 or 6 vertices. For an illustration, consider \cref{fig:suff-0-a}. 
    In other words, there exists a decomposition into rectangles and L-shapes. 
    For each of these parts one \rguard{0} suffices: 
    In an L-shape, we may place it in the cell with corner contact to the reflex corner, for rectangles a guard can be placed at any cell. 
\end{claimproof}

\begin{claim}
    \label{lem:suff-nec-0-m}
    $\floor{\nicefrac{m}{2}}$  \rguards{0} are always sufficient.
\end{claim}

\begin{claimproof}
    Note that every \rguard{0} sees its own and its adjacent cells. 
    As the dual graph of a polyomino is bipartite, we may color the vertices in white and black such that no two cells of the same color share an edge, see \cref{fig:suff-0-b}. 
    Hence, placing \rguards{0} on all cells of one color yields a valid \rguard{0} set. 
    Clearly, one of the two color classes contains at most $\floor{\nicefrac{m}{2}}$ cells which proves the claim.
\end{claimproof}

This completes the proof.
\end{proof}

\head{\subsection{\rguards{1}}}
We now prove a tight bound on the number of \rguards{1} that are sometimes necessary and always sufficient to guard a hole-free polyomino.

\hypertarget{proof:onerestlessnecessarysufficient}{}
\onerestlessnecessarysufficient*

\begin{proof}
We again prove the two parts separately.

\begin{claim}
    \label{lem:suff-nec-1-n}
    $\nicefrac{(n-2)}{2}$ guards are always sufficient and sometimes necessary.
   
\end{claim}
\begin{claimproof}
    \Cref{fig:necessary-1A-a} depicts polyominoes that require at least $\nicefrac{(n-2)}{2}$ \rguards{1}.
    They consist of smaller windmill-shaped polyominoes.
    In a windmill, a \guard{1} must be placed on each of the four extremal cells; 
    otherwise these cells are not guarded. 
    Then, an additional \rguard{1} must be placed on the central cell, which remains unguarded otherwise.
    By connecting multiple windmills along one of the teeth, we can construct arbitrarily large polyominoes while preserving this structural property. 
    \begin{figure}[htb]
    	\centering
    	\captionsetup[subfigure]{justification=centering}
    	\begin{subfigure}{0.60\columnwidth}%
    		\centering
    		\includegraphics[page=1]{necessary_n}
    		\subcaption{}
    		\label{fig:necessary-1A-a}
    	\end{subfigure}%
    	\begin{subfigure}{0.40\columnwidth}%
    		\centering
    		\includegraphics[page=2]{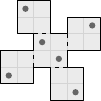}
    		\subcaption{}
    		\label{fig:necessary-1A-c}
    	\end{subfigure}%
    	\caption{
            (a) A family of polyominoes that require $\nicefrac{(n-2)}{2}$ \rguards{1} to be entirely guarded.
            (b) A polyomino partitioned into rectangles where more than one guard is required in one rectangle.
            }
    	\label{fig:necessary-1A}
    \end{figure}

	In all but the first windmill, the ratio between the number of guards and vertices is $1:2$. 
	In the first windmill, we have two additional vertices. 
	This gives the desired bound. 
	
	\medskip
    It remains to show sufficiency. 
    To this end, we first partition the polyomino $P$ into at most $\nicefrac{(n-2)}{2}$ rectangles. 
    In a second step, we carefully place one \guard{1} in each rectangle. 
    Lastly, we show that the guard set is valid. 
    We highlight that an arbitrary partition into rectangles does not have the claimed properties; 
    for example, consider the polyomino depicted in \cref{fig:necessary-1A-c} with a partition into five rectangles. 
	In this example, no valid guard set contains only one \guard{1} in the central rectangle. 
    Hence, starting from a  good partition is crucial. 
   
	The idea of the partition is to iteratively grow new rectangles from so called \emph{open edges}. 
	Throughout the procedure, let $P' \subseteq P$ denote the union of already constructed rectangles.
	We start with a rectangle $r_0$ that is an arbitrary maximal rectangle in $P$ sharing one of its edges with an \outerdent of $P$; we call the \outerdent the \emph{open edge} of $r_0$. 
	We partition the edges of $r_0$ with vertices from $P$; note that these vertices have \emph{flat} angles in $P'$.
	We call the edges of $P'$ not belonging to $P$ \emph{open}. 
    Next, we iteratively select an open edge $e$ and construct a maximal rectangle $r$ inside $P \setminus P'$ with edge $e$. 
    In particular, the side of $r$ opposite to $e$ touches the boundary of $P$. 
    We identify open edges of~$r$ and proceed until $P'=P$. 
	For an example of the construction consider \cref{fig:suff-1-n-overview}.
    
	\begin{figure}[htb]
    	\centering
    	\captionsetup[subfigure]{justification=centering}
    	\begin{subfigure}{0.7\textwidth}%
    		\centering
    		\includegraphics[page=4]{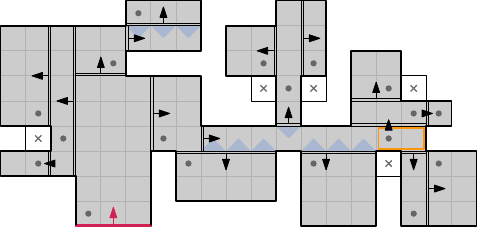}
    		\subcaption{}
    		\label{fig:suff-1-n-overview-a}
    	\end{subfigure}%
    	\begin{subfigure}{0.3\textwidth}%
    		\centering
    		\includegraphics[page=5]{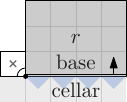}
    		\subcaption{}
    		\label{fig:suff-1-n-overview-b}
    	\end{subfigure}%
    	\caption{
    		Constructing a set of at most $\nicefrac{(n-2)}{2}$ \rguards{1}. 
    		(a)
    		Starting from a maximal rectangle containing the \outerdent (red), $P'$ is iteratively extended from open edges; double lines not only show the partition into rectangles, but also indicate open edges during the procedure. 
    		A~guard is placed in every rectangle;
    		in rectangles of width 1, blue triangles indicate cells that are already seen as a cellar of an adjacent rectangle.
    		(b) A rectangle $r$ with base and cellar.
    		}
    	\label{fig:suff-1-n-overview}
    \end{figure}
    
    We now argue that in each step, two new edges 
    of $P$ are \emph{encountered}, i.e., $P'$ does not yet touch these sides. 
    We note that at least one vertex of $e$ is a vertex of $P$ and hence incident to an edge of $P$. 
    Moreover, the edge opposite of $e$ is incident to an edge of $P$. 
    If any of the two edges had been encountered before, we obtain a contradiction to hole-freeness of the added rectangles.
    Hence, after $s$ steps, $P'$ is a $(2+2s)$-gon; 
	that is, we repeat the procedure until $n=2s+2\iff s=\nicefrac{(n-2)}{2}$. 
	Consequently, we obtain a partition with at most $\nicefrac{(n-2)}{2}$ rectangles, as desired.

    We now explain how to place our guards.
	To this end, we consider a rectangle $r$ from $R$ whose defining open edge is without loss of generality at the bottom.
	We denote the cells of~$r$ in the bottommost row as its \emph{base}, and the cells right below the base as its \emph{cellar}, see \cref{fig:suff-1-n-overview-b}. 
	We first note that the initial rectangle $r_0$ can be guarded from a position on its base; in the case of height 1, the guard must be placed below an outer cell, otherwise any position works.
    Now, let $r$ be any other rectangle in the partition $R$. 
    As open edges share a vertex with $P$, there exists an outer cell adjacent to one of the vertical edges right above the base, as in \cref{fig:suff-1-n-overview-b}. 
	By construction, the cellar of $r$ is contained in $P'$.
    If $r$ has width and height at least $2$, we place a \rguard{1} on the base next to the outer cell, as in \cref{fig:suff-1-n-a}.
	Note that it guards $r$ and its cellar. 
    If $r$ has a width at least $2$ and a height of~$1$, we place a guard in $r$ below an outer cell incident to the top of $r$; 
	note that such an outer cell exists as otherwise $r$ is not maximal.
	For an example consider \cref{fig:suff-1-n-b}. 
	Clearly, such a guard sees $r$ and in most cases also the cellar.
	If the guard does not see the cellar of $r$, then the cellar contains a topmost cell in the cellar's rectangle, see \cref{fig:suff-1-n-b2}.
    
	\begin{figure}[htb]
    	\centering
    	\captionsetup[subfigure]{justification=centering}
    	\begin{subfigure}{0.23\textwidth}%
    		\centering
    		\includegraphics[page=6]{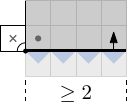}
    		\subcaption{}
    		\label{fig:suff-1-n-a}
    	\end{subfigure}%
    	\begin{subfigure}{0.23\textwidth}%
    		\centering
    		\includegraphics[page=7]{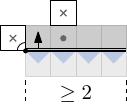}
    		\subcaption{}
    		\label{fig:suff-1-n-b}
    	\end{subfigure}%
    	\begin{subfigure}{0.24\textwidth}%
    		\centering
    		\includegraphics[page=8]{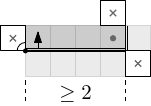}
    		\subcaption{}
    		\label{fig:suff-1-n-b2}
    	\end{subfigure}%
    	\begin{subfigure}{0.15\textwidth}%
    		\centering
    		\includegraphics[page=9]{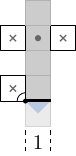}
    		\subcaption{}
    		\label{fig:suff-1-n-c}
    	\end{subfigure}%
    	\begin{subfigure}{0.15\textwidth}%
    		\centering
    		\includegraphics[page=10]{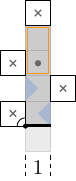}
    		\subcaption{}
    		\label{fig:suff-1-n-d}
    	\end{subfigure}%
    	\caption{Guard placements for all possible types of rectangles.
    	}
    	\label{fig:suff-1-n}
    \end{figure}
  
    It remains to guard $r$ of width 1. 
    If $r$ contains a cell that is adjacent to outer cells to both its left and right, we place a guard there, as shown in \cref{fig:suff-1-n-c}. 
    Clearly, this guard sees~$r$ and its cellar. 
	Otherwise, we place a guard in the topmost cell of $r$ that has an outer cell to its left or right, see \cref{fig:suff-1-n-d}.
	This ensures that all cells in $r$ at this position and above are guarded.
	Note that in this case, the cellar of $r$ is a topmost cell. 

	It remains to show that $r$ is guarded for this last case. 
	To this end, we note that each cell of $r$ is in some cellar of an adjacent rectangle. 
	Moreover, all cells of $r$ are guarded as the cellar of their adjacent rectangles, except for possibly a sequence of topmost cells, as for example the orange cells in \cref{fig:suff-1-n-overview-a}.
	The guard placed in $r$ takes care of guarding these cells.	
    This concludes the proof.
\end{claimproof}

\begin{claim}
    \label{lem:suff-nec-1-m}
    For a polyomino with at least $m\geq 4$ cells, $(m-2)$ \rguards{1} are always sufficient and sometimes necessary.
\end{claim}

\begin{claimproof}
    Staircase polyominoes, as visualized in~\cref{fig:necessary-1A-staircase}, require $(m-2)$ \rguards{1}.
    The cells at both `ends' are only seen by a guard that is placed on them; these guards also see their respective adjacent cell. Moreover, all other cells require their own \rguards{1} placed on them.

	\begin{figure}[htb]
		\centering
		\includegraphics[page=2]{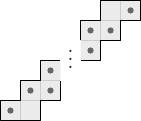}
		\caption{Staircase polyominoes require $m-2$ guards for a complete \guarding{1}.}
		\label{fig:necessary-1A-staircase}
	\end{figure}

    We now argue that $m-2$ many \rguards{1} are always sufficient to see a polyomino. 
    No matter its position, every \rguard{1} sees the cell where it is initially placed on. 
    Therefore, it suffices to identify two \rguards{1} that can guard two cells.
    
    We first assume that $P$ has a bounding box of height (or width) $\geq 4$. 
    \Cref{fig:suff-1-m-a} shows that in two topmost rows of $P$ such a \rguard{1} can be found and that the additionally guarded cell is also contained in the first two rows. 
    With a polyomino height of $\geq 4$, the two topmost and two bottommost rows are distinct. 
    As \cref{fig:suff-1-m-a} equally works for the bottommost rows, we can identify two distinct \rguards{1} that also guard two more distinct cells.
    It remains to investigate the cases where the bounding box is at most $3 \times 3$. 
    A~closer look reveals that only in the first (vertical) case in \cref{fig:suff-1-m-a} the adjacent guarded cell lies in the second row. 
    Hence, the same argument works for height (or width) $\geq 3$ when the vertical case occurs at most once or if it identifies two different cell positions. 
    In turn, if the same cell is identified twice in the vertical case, $P$ either is a plus or a rotated `T', and we show in \cref{fig:suff-1-m-b} how to guard these. 
    Lastly, if the bounding box of $P$ is a $2\times 2$-square, a single guard suffices because $m\geq 4$.
\end{claimproof}

\begin{figure}[htb]
	\centering
		\captionsetup[subfigure]{justification=centering}
	\begin{subfigure}{0.6\textwidth}%
		\centering
		\includegraphics[page=3]{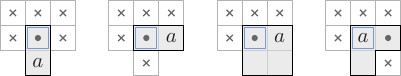}
		\subcaption{}
		\label{fig:suff-1-m-a}
	\end{subfigure}%
	\begin{subfigure}{0.4\textwidth}%
		\centering
		\includegraphics[page=4]{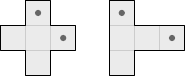}
		\subcaption{}
		\label{fig:suff-1-m-b}
	\end{subfigure}%
	\caption{
		(a) Based on the top left cell in a polyomino $P$ (blue outline), we find a \rguard{1} that also guards an adjacent cell $a$. 
		We show all possible configurations around this cell in~$P$. 
		Cells marked with a cross are known to be outer cells. 
		(b) If $P$ lies inside a $3 \times 3$ bounding box, the depicted special cases can be guarded as shown.}
	\label{fig:suff-1-m}
\end{figure}

This completes the proof.
\end{proof}

\section{Omitted details for~\texorpdfstring{\cref{sec:complexity}}{Section~\ref{sec:complexity}}}

We now provide the details of the \NP-completeness
for polyominoes with holes. 
More precisely, we show that for any $k\in\mathbb{N}$, deciding whether a given polyomino can be guarded with $\ell$ many \rguards{k} is~\NP-complete.

On a high level, our reduction for the \NP-hardness is an extension of the  approach described by Biedl and Mehrabi~\cite{Biedl016} for guarding polyominoes with point guards under rectangular visibility.
They show that finding a minimum vertex cover remains \NP-complete if the graph is orthogonally embedded into the plane, all vertices lie at integer coordinates and the maximum degree is three. 
Based on such an embedding, they construct a polyomino $P$ that is basically a ``thickened'' version of the embedding of $G$.
More precisely, it contains corridors of unit height or width for every edge in the given graph $G$, and the corridors are connected by gadgets that depend on the vertex degree. 
They then show that a minimum vertex cover in $G$ corresponds to a minimum point guard set of $P$.

Our reduction follows the same strategy: 
We reduce from the aforementioned \NP\nobreakdash-complete orthogonal \textsc{Vertex Cover} variant~\cite{Biedl016}, constructing a polyomino whose corridors are connected via three types of gadgets determined by the vertex layout.
Based on the number and direction of incoming edges, a vertex can either be a \emph{corner}, a \emph{tunnel}, or a \emph{junction}. 
Note that the reduction holds for all $k\in \mathbb{N}$, including $k=0$, provided the instance is scaled such that every corridor has width or height $2k+1$. 
All gadgets are visualized in~\cref{fig:hardnessA}. 

Building on these insights and the gadgets depicted in~\cref{fig:hardnessA}, we now establish the following result.

\hypertarget{proof:propHardness}{}
\propHardness*

\begin{proof}
We begin by showing \NP{} membership. 
Let $P$ be a polyomino with $n$ vertices.
We aim to check in time polynomial in $n$ and $\ell$ whether a given set of at most $\ell$ candidate \rguards{k}  sees the entire polyomino $P$. (In fact, \cref{claim:char} and  \cref{claim:partition}(\ref{item:i}) imply that  if  $P$ is guardable, then there exists a guard set of size at most $n^2$.)
To this end, we compute \viska{k}{g} for every guard~$g$ and check whether their union covers $P$.

Let $R$ be the partition of $P$ into rectangles as described in \cref{prop:existencePolytime}.
By \cref{claim:partition}(\ref{item:i}), the visibility polyominoes of all positions in a rectangle in $R$ are the same.
Thus, for each rectangle $r$ of $R$ intersecting the $k$-neighborhood of $g$, we compute  $\viska{0}{r}$ which is a union of all rectangles $r'$ from $R$ for which $\rect(r,r')\subseteq P$. Then we obtain $\viska{k}{r}$ by taking their intersection.
Whether a rectangle of $R$ intersects the $k$-neighborhood can be determined independent of the size $k$.
Lastly, we check if the union of $\viska{k}{g}$ for all guards covers $P$.
We remark that all intersection and union operations can be performed efficiently because by \cref{claim:partition}(\ref{item:ii}) all visibility polyominoes only consist of rectangles in $R$. 
This concludes the membership proof.

Now, we turn our attention to the \NP-hardness.
We demonstrate that the \NP-hard problem of finding a vertex cover in an orthogonal embedding of a graph $G$ reduces to finding a \rguard{k} set in a polyomino $P$. 
The polyomino $P$ is obtained by thickening $G$, using the gadgets from~\cref{fig:hardnessA}. 
Throughout the proof, we refer to various regions of $P$ as indicated and labeled in the figure.

\begin{figure}[htb]
	\centering
	\captionsetup[subfigure]{justification=centering}
	\begin{subfigure}{0.18\textwidth}%
		\centering
		\includegraphics[page=1]{hardness_pages}
		\subcaption{}
		\label{fig:hardnessA-a}
	\end{subfigure}%
	\begin{subfigure}{0.15\textwidth}%
		\centering
		\includegraphics[page=2]{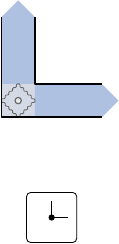}
		\subcaption{}
		\label{fig:hardnessA-b}
	\end{subfigure}%
	\begin{subfigure}{0.25\textwidth}%
		\centering
		\includegraphics[page=3]{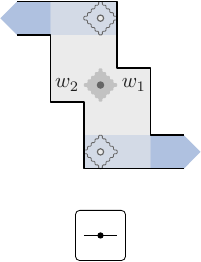}
		\subcaption{}
		\label{fig:hardnessA-c}
	\end{subfigure}%
	\begin{subfigure}{0.42\textwidth}%
		\centering
		\includegraphics[page=5]{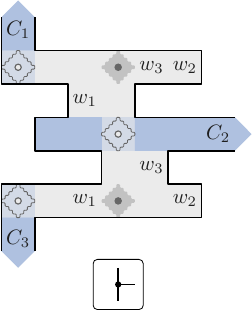}
		\subcaption{}
		\label{fig:hardnessA-d}
	\end{subfigure}%
	\caption{
		The gadgets used for reducing from the minimum vertex cover problem.
		The corridors connect the other gadgets.
		(a)
		A~corridor. 
		(b)
		A~corner gadget.
		(c)
		A~tunnel gadget.
		(d)~A~junction gadget.
		For all gadgets replacing a vertex of $G$, we depict two different guard placements. 
		Gray guards are sufficient if all connecting corridors are already guarded;
		otherwise, the outlined guards see the respective gadget, as well as all corridors connected to it.}
	\label{fig:hardnessA}
\end{figure}

The analysis of our reduction is slightly more involved than the reduction by Biedl and Mehrabi~\cite{Biedl016} because some of our gadgets need a `ground set' of guards that are always required;
we refer to the cardinality of this set as $\minguard$. 

Note that this number of guards only suffices if all incoming corridors are already guarded from somewhere else, otherwise $\minguard+1$ guards are required. 

\newcommand{\gad}{\ensuremath{\Gamma}\xspace}

More formally, let \gad be a gadget in $P$. 
We use $\minguard$ to assign a value to every vertex gadget based on the gadget type:
\[
\minguard(\gad) := \begin{cases} 
    0, \text{ if \gad is a corner gadget,}\\
    1, \text{ if \gad is a tunnel gadget,}\\
    2, \text{ if \gad is a junction gadget.}
\end{cases}
\]

We define $\offs$ as the sum $\sum_{\gad} \minguard(\gad)$ of all gadget values in $P$. 
We split the remaining proof into several claims.

\begin{claim}
    \label{clm:hardness-iif}
    There exists a vertex cover of size $\lambda$ for $G$ if and only if the polyomino $P$ can be guarded with $\lambda+\offs$ many \rguards{k}.
\end{claim}

In \cref{fig:hardnessA}, the gray gadgets overlap with the blue corridors in light blue rectangular regions. 
We refer to the gray part of the gadgets as the gadget's \emph{core}, to the light blue regions as the \emph{overlap}, and we call the union of both the \emph{gadget}.
We say that a \rguard{k} is \emph{fully contained within} a polyomino $P'\subseteq P$ if its $k$-neighborhood is a subset of $P'$. 

\begin{claim}
    \label{clm:hardness-place-guards}
    Every gadget $\gad$ can be guarded by $\minguard(\gad)+1$ many \rguards{k}, and its core by  $\minguard(\gad)$.
\end{claim}
\begin{claimproof}
    Corresponding guard placements are indicated in~\cref{fig:hardnessA} for the three gadgets; 
    there are $\minguard(\gad)$ many \rguards{k} that see the gadget's core (highlighted in dark gray), and $\minguard(\gad)+1$ outlined guards that see the gadget's core as well as all connecting corridors.\claimqedhere
\end{claimproof}

\begin{claim}
    \label{clm:hardness-count-min}
    Let \gad be a gadget in $P$ and let $S$ be a \rguard{k} set for $P$. 
    Then~$S$ contains at least $\minguard(\gad)$ guards that are fully contained within \gad. 
\end{claim}
\begin{claimproof}
    The statement is clear for corner gadgets, as $\minguard(\gad)=0$. 
	For the other gadgets, refer to \cref{fig:hardnessA} for the cells labeled $w_1$ and $w_2$. 
    In tunnel gadgets, we note that $\viska{0}{w_1}$ is contained within \gad, and thus a \rguard{k} that sees it must be fully contained within it as well. 
    For junction gadgets, similar reasoning applies: 
	The visibilities of both cells labeled~$w_2$ have the same property. 
    Additionally, as both visibility polyominoes are distinct, two separate \rguards{k} are required.
\end{claimproof}

\begin{claim}
    \label{clm:hardness-corridor}
    Let $C$ be a corridor in $P$, and let $S$ be a \rguard{k} set for $P$. 
    Then there exists a \rguard{k} in $S$ that sees $C$.
	Also, every guard that sees $C$ is fully contained within $C$.
\end{claim}
\begin{claimproof}
    By construction, there is a cell $c$ in the corridor $C$ that is not part of any gadget~overlap. 
    As $\viska{0}{c}$ coincides with $C$, it implies that if a \rguard{k} $g$ sees $c$ (and at least one guard in $S$ must do so), its $k$-neighborhood~$N_k(g)$ is contained in $C$.
    Moreover, since $C$ is a rectangle, every cell is seen from any position within~$C$, implying that $g$ sees $C$. 
\end{claimproof}

\begin{claim}
    \label{clm:hardness-count-guards}
    Let \gad be a gadget in $P$ and let $S$ be a \rguard{k} set for $P$. 
    Furthermore, let $S' \subseteq S$ be the subset of guards that are fully contained within \gad. 
    Then, if a guard $g\in S'$ sees a corridor, we have $|S'| \geq \minguard(\gad)+1$. 
\end{claim}
\begin{claimproof}
    The statement is trivially true for corners: 
	Once any guard is fully contained in a corner, we already have $\minguard(\gad)+1$ guards.
    
    For tunnels and junctions, we observe that, by~\cref{clm:hardness-corridor}, the guard $g$ is fully contained in a (light-blue) overlap.
    In tunnels, either $w_1$ or $w_2$ remains unseen by $g$; 
	consequently, a second \rguard{k} is necessary to guard the gadget~\gad. 
    Similarly, in junctions there exist two cells that are not seen by $g$ and whose visibility polyominoes are distinct, implying that they cannot be covered by the same \rguard{k}. 
    In \cref{fig:hardnessA}, we indicate such possible cell positions; when $g$ sees the corridor $C_i$, the two associated cells are labeled $w_i$. 
\end{claimproof}    

We are now ready to prove \cref{clm:hardness-iif}. 
\begin{claimproof}[Proof of \cref{clm:hardness-iif}.]
    We show both directions separately. 
    Let $J$ be a vertex cover of size $\lambda$ for~$G$. 
    We use~\cref{clm:hardness-place-guards} to guard $P$ by placing $\minguard(\gad)+1$ many \rguards{k} for every vertex that is part of the vertex cover $J$, and $\minguard(\gad)$ many \rguards{k} for the other vertices to guard the corresponding gadgets; in total, this guards the entire polyomino.

    Now, let $S$ be a \rguard{k} set for $P$. 
    Given $S$, we construct a vertex cover $J$ for~$G$ as follows. 
    As argued in~\cref{clm:hardness-corridor}, every corridor in $P$ fully contains a \rguard{k}~$g$. 
    We distinguish two cases: (1) If $g$ is fully contained in the overlap of a gadget, we add the associated vertex to the vertex cover $J$.
    (2) Otherwise, we arbitrarily choose one of the incident vertices in $G$, and add it to $J$. 
    After considering all corridors, $J$ necessarily constitutes a vertex cover of $G$.
    
    It remains to show that $|J|+\offs \leq |S|$. 
    To this end, we apply \cref{clm:hardness-count-min} to account for $\minguard(\gad)$ guards contained in every gadget \gad in $P$. 
    Altogether, this identifies $\offs$ many different guards in $S$. 
    What remains is to identify another \rguard{k} for every vertex in $J$.
    For this, consider the following case distinction:
    If $v \in J$ was added in Case (1), it follows that a guard was fully contained in the overlap of the corresponding gadget.
    In this case, \cref{clm:hardness-count-guards} ensures that there is another \rguard{k} $g$ fully contained in the gadget, and this guard has not yet been accounted for in $\offs$. 
    Otherwise, if $v$ arose from Case (2), then the guard in question was located within a corridor, and was therefore not fully contained in a gadget.
    Clearly, $\offs$ does not include this guard either. 
    Repeating this reasoning for every vertex in $J$ yields at least $|J|$ additional \rguards{k} in $S$.
    Therefore,~$S$ contains at least $|J|+\offs$ \rguards{k}. 
\end{claimproof}

\cref{clm:hardness-iif,clm:hardness-place-guards,clm:hardness-corridor,clm:hardness-count-min,clm:hardness-count-guards} together with the aforementioned construction of $P$ from a given embedding of $G$ concludes the proof.
\end{proof}

\section{Omitted details for~\texorpdfstring{\cref{sec:k-restless}}{Section~\ref{sec:k-restless}}}

\label{sec:k-guards}
In this section, we present the details for the proofs in \cref{sec:k-restless}, namely for \cref{lem:separation-better}, \cref{thm:k-guard-opt-hist}, \cref{prop:0-approx-2,lem:5}.

\subsection{A crucial tool for lower bounds}
We now present the general proof for \cref{lem:separation-better} that not only works in histograms, but also for embedded histograms. 

\hypertarget{proof:separation}{}
\separation*

\begin{figure}[htb]
	\centering
	\includegraphics[page=2]{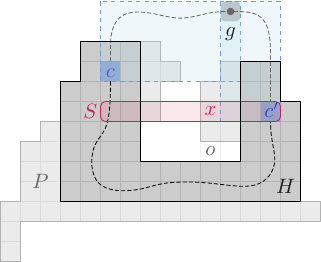}
	\caption{Illustration for the proof of \cref{lem:separation-better}: 
		If $S$ contains a cell $x$ not contained in the histogram $H$ and a guard $g$ sees both $c$ and $c'$, we can find an outer cell $o$ that is enclosed by paths of cells in $P$ both above and below, implying that $P$ is not hole-free.}
	\label{fig:separation-embedded}
\end{figure}

\begin{proof}
	Without loss of generality, we assume that $H$ is upwards, as visualized in~\cref{fig:separation-embedded}.
	For the sake of contradiction, suppose that there exists a cell $x$ in $S$ that is not contained in~$H$.  
	We identify an outer cell $o$ of $P$ not above $x$ as follows: 
	If $x$ is an outer cell, we set~$o:=x$. 
	Otherwise, we pick the first outer cell of $P$ directly below $x$. 
	Since the top \outerdents of an embedded histogram are adjacent to outer cells, the outer cell $o$ always lies above~$H$. 
	
	Note that $g$ is placed in a position above $o$; otherwise $o$ blocks the vision to  $c$ or $c'$. 
	Because $g$ sees both $c$ and $c'$, there exists a path inside $P$ from $c$ over~$g$ to~$c'$ through $\rect(c,g)$ and $\rect(g,c')$ above $o$. 
	However, $c$ and $c'$ are also connected by a path through the histogram which lies below $o$. 
	Therefore, $o$ is enclosed by cells of $P$ which implies that $P$ has a hole, a contradiction to the hole-freeness of $P$.
\end{proof}

\subsection{On optimally \guarding{k} histograms}

We now provide the details on optimally \guarding{k} histograms. 

\hypertarget{proof:theoremOptHist}{}
\theoremOptHist*

\begin{proof}
	We start by explaining our strategy. 
	First, we partition $H$ into maximal horizontal strips and merge all vertically adjacent strips of equal width. 
	This yields a partition of $H$ into rectangles, which we refer to as \emph{blocks}. 
	For an illustration, consult \cref{fig:blocksA}. 
	Note that the number of blocks is at most $\nicefrac{n}{2}$; the top edge of each block is incident to at least two vertices.
    Because blocks are enclosed by outer cells to the left and right, blocks are either disjoint or \emph{vertically adjacent}. 
	These adjacencies yield an \emph{adjacency graph} $T$ of all blocks in~$H$, where edges are directed upwards. 
    Note that by the histogram structure, $T$ is a rooted tree, and the root corresponds to the block containing the histogram base. 
	The \emph{extension} of a block~\block is the rectangle $\block'$ obtained by extending~\block vertically to the histogram base.
    In every block, we select one cell that is adjacent to an outer cell to its top. 
    We refer to this cell as the block's \emph{key} cell. 
    
    \begin{figure}[htb]
    	\centering
    	\includegraphics{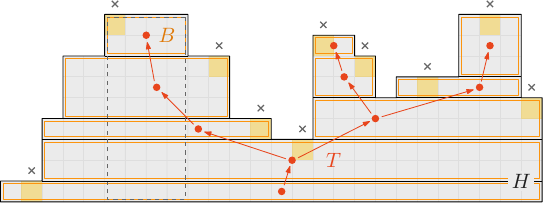}
    	\caption{
    		Optimally \guarding{k} a histogram $H$ is based on a partition into rectangular blocks (highlighted in orange) that form a rooted tree $T$. 
        	The extension of a block reaches down to the histogram base, as outlined for block $B$ in dashed lines.
        	Each block is assigned a key cell within the respective block, visualized in yellow.}
    	\label{fig:blocksA}
    \end{figure}

We establish the following key lemma.
    \begin{lemma}
	\label{clm:k-guard-hist-path-vis}
	The blocks seen by a \rguard{k} induce a directed path in $T$. 
\end{lemma}

We show the lemma with the help of three claims.
    \begin{claim}
		\label{clm:k-guard-hist-ext}
		If a \rguard{k} $g$ sees the key cell $a$ of block \block, then its $k$-neighborhood~$N_k(g)$ lies  in the extension $\block'$ of \block. 
	\end{claim}
	\begin{claimproof}
        We notice that $\viska{0}{a}=\block'$, as \block is enclosed by outer cells to the left and right and the adjacent outer cell above $a$ restricts the visibility to the top. 
        Let $c\in N_k(g)$. Because $g$ sees $a$, we have $a\in \viska{0}{c}$ and by symmetry $c\in \viska{0}{a}=\block'$.
	\end{claimproof}

	\begin{claim}
		\label{clm:k-guard-hist-iif}
		A \rguard{k} $g$ sees a block \block if and only if it sees its key cell $a$. 
	\end{claim}
	\begin{claimproof}
		Clearly, if $g$ sees \block, this includes $a \in \block$. 
		If $g$ sees $a$ then, by \cref{clm:k-guard-hist-ext}, $N_k(g)$ is contained in the extension $\block'$ of $\block$. 
		Because $\block'$ is contained in $H$, for each cell $c\in N_k(g)\subseteq \block'$, we have $\block\subseteq \viska{0}{c}$. 
		Thus $g$ sees each cell of \block. 
	\end{claimproof}

    \begin{claim}
		\label{clm:k-guard-hist-path-exists}
        If a \rguard{k} $g$ sees two different blocks $B_1$ and $B_2$, there exists a directed path $\rho$ in~$T$ between them. 
		Moreover, $g$ sees all blocks on $\rho$.
	\end{claim}

    \newcommand{\north}{1}
    \newcommand{\south}{2}
    
	\begin{claimproof}
		We illustrate the proof in \cref{fig:histogram-details}. 
		Let $a_i$ denote the key cell of $\block_i$. 
		By \cref{lem:separation-better},~$a_1$ and~$a_2$ have different vertical positions; 
		otherwise $a_1$ and $a_2$ belong to the same block, a contradiction.
		Let $a_{\north}$ denote the cell further to the top. 
		Let $c$ be the cell at the $x$-position of~$a_{\north}$ and the $y$-position of~$a_{\south}$. 
        By  \cref{lem:separation-better}, the horizontal strip from $c$ to $a_{\south}$ is contained in $H$. 
        This means that $c$ also lies in~$\block_{\south}$. 
		Moreover, because $H$ is a histogram, the vertical strip $V$ from $c$ to $a_{\north}$ is contained in~$H$. 
        Consequently, there exists a sequence of blocks intersecting~$V$ that are vertically adjacent and therefore form a directed path $\rho$ from~$\block_{\south}$ to~$\block_{\north}$ in $T$. 

		It remains to argue that $g$ sees each block $\block$ on $\rho$. 
		By \cref{clm:k-guard-hist-ext}, $N_k(g)$ lies in the intersection of both extended blocks $\block_{\north}' \cap \block_{\south}'$. 
        By the structure of histograms and the fact that~$\block$ is a block between $\block_{\north}$ and $\block_{\south}$, $\block_{\north}' \cap \block_{\south}'$ is contained in the extension of $\block$. 
		Therefore,~$g$ sees every cell of $\block$. 
    \end{claimproof}
	\begin{figure}[htb]
    	\centering
    	\includegraphics[page=2]{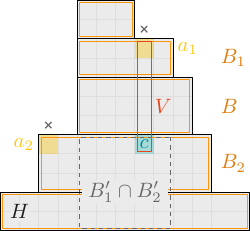}
    	\caption{
    		Proving \cref{clm:k-guard-hist-path-exists}: 
			If a \rguard{k} sees the blocks $\block_1$ and $\block_2$, we know that the strips that connect the respective key cells $a_1$ and $a_2$ with $c$ are contained in $H$, implying that there is a directed path in $T$.
			Also, the $k$-neighborhood of the guard must be contained in $\block_{\north}' \cap \block_{\south}'$.}
    	\label{fig:histogram-details}
    \end{figure}

    Now, \cref{clm:k-guard-hist-path-exists} immediately implies the lemma.

    \begin{proof}[Proof of \cref{clm:k-guard-hist-path-vis}]
		\cref{clm:k-guard-hist-path-exists} implies that any two seen blocks are in child-parent relationship. 
		Moreover, all blocks on their path are seen. 
        Note that this also covers the cases where a \rguard{k} only sees a single block or no block in $H$. 
    \end{proof}

    We now explain how to compute a \rguard{k} set for a histogram $H$.
	By \cref{clm:k-guard-hist-path-vis}, guarding $H$ translates to covering the rooted tree $T$ by a minimum number of paths induced by all \rguards{k}. 
	To this end, we first identify the set of paths in $T$. 
	We can test the existence of a path between two blocks $B_1$ and $B_2$ of $T$ in child-parent relationship as follows. 
	Let $a_1$ and $a_2$ be the key cells of $B_1$ and $B_2$, respectively. 
	By \cref{lem:guardable}, we can check whether~$a_1$ and $a_2$ can be seen by a \rguard{k} in polynomial time, and identify a suitable guard if one exists. 
	This gives a trivial upper bound of at most $\bO(|T|^2)\subseteq \bO(n^2)$ pairs to check. 
	As we are only interested in maximal paths, this can be pruned, e.g., to~$\bO(n\log n)$ pairs with binary search. 
	Hence, we may assume that we are given a tree $T$ and a polynomial number of paths in $T$.
   	
	Now that we identified all directed paths in $T$, we can cover $T$. We use the following strategy: 
    In $T$, we mark each block to be either guarded or unguarded. Initially, all blocks are unguarded. 
    We also maintain an initially empty guard set $G$ and an initially empty set of key cells $A$. 
    While it exists, we select an unguarded block $\block$ in $T$ whose children are already guarded (or $\block$ is a leaf). 
    Then, we choose a path seeing \block which comes closest to the root of~$T$.
    We add the corresponding guard~$g$ to~$G$ and the key cell of~$\block$ to~$A$.
    If no such path exists, $\block$ was not guardable, and hence we return that $H$ cannot be \guarded{k}. 
	We mark all blocks on that path as guarded.

    It remains to show correctness.
    The fact that $G$ is a valid guard set follows from \cref{clm:k-guard-hist-path-vis} and the fact that by construction, each vertex of $T$ is covered by a path.
    
	For minimality, we show that no two key cells in~$A$ can be guarded by the same \rguard{k}, i.e., $A$ is a \emph{witness} set. 
	Then, any guard set requires at least $|A|$ guards and the fact that $|G|=|A|$ implies optimality.
	
	For the sake of a contradiction, suppose there exists a \rguard{k} $g$ that sees two different key cells $a$ and $a'$ in~$A$. 
	By \cref{clm:k-guard-hist-path-exists}, $a$ and $a'$ lie on a directed path in $T$. 
	Without loss of generality, let $a$ denote the key cell of the block closer to the root in $T$. Note that $a'$ was added to $A$ before $a$ because it lies in a subtree of $a$. 
	When handling $a'$, we added a guard $g'$ to $G$ that sees $a'$ and maximizes the number of seen blocks in root direction. 
	As $a$ is added to $A$, it remains unguarded after $g'$ is added.
	However, $g$ sees both $a$ and $a'$ and therefore reaches further to the root than $g'$. 
	Therefore, $g$ would have been chosen instead of~$g'$, a contradiction. 
    This proves the theorem.
\end{proof}

\subsection{A 2-approximation for \guarding{0} polyominoes}

We now show how a simple observation allows us to approximate \rguard{0} sets for hole-free polyominoes. 

\hypertarget{proof:twoRestlessOnePoint}{}
\twoRestlessOnePoint*
\begin{proof}
	If the point guard $p$ lies in the interior of a cell, we replace it by a \rguard{0}~$g$ at the same location; 
	note that their visibility regions coincide. 
	If $p$ lies at the boundary of a cell, we split the visibility region of $p$ along the bisector coinciding with the boundary; 
	without loss of generality vertical. 
	Because the visibility polygons are star-shaped, the resulting polyominoes are histograms that have a single left or right \outerdent, respectively. 
	In \cref{fig:two-approxA} these are marked with a bold line. 
	
	\begin{figure}[htb]
		\centering
		\includegraphics[page=2]{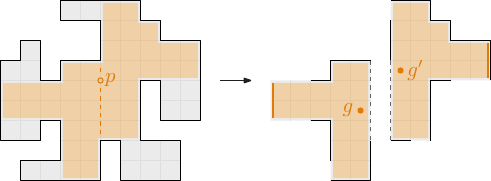}
		\caption{The visibility region of a point guard can be split into two histograms that can be guarded with a single \rguard{0} each.
		}
		\label{fig:two-approxA}
	\end{figure}
	
	Placing a \rguard{0} on the base of each histogram such that it sees the \outerdent guards the respective histogram. 
	Hence, the visibility region of $p$ is covered by two \rguards{0}.
	Note that the same arguments hold if $p$ lies on a grid point.
\end{proof}

\subsection{From embedded histograms to hole-free polyominoes}

We provide details on combining guard sets for embedded histograms to entirely guard hole-free polyominoes, while introducing an additional approximation factor of five. 

\hypertarget{proof:fiveapproxlemma}{}
\fiveapproxlemma*

\begin{proof}
This lemma extends a result of Lingas, Wasylewicz, and \.Zyli\'nski~\cite{3approx}, which is in turn based on a linear-time partition of hole-free orthogonal polygons into a set of histograms presented by Levcopoulos~\cite{levcopoulos}, as described in \cref{thm:k-guard-poly} and visualized in~\cref{fig:5A}.
	
	\begin{figure}[htb]
		\centering
		\includegraphics[]{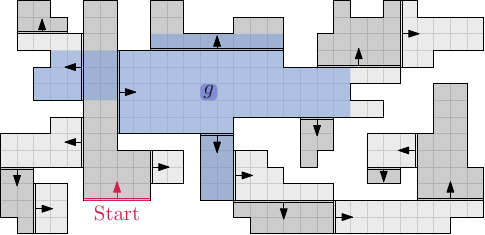}
		\caption{
			Partitioning a polyomino into histograms, with red being the starting edge. 
		}
		\label{fig:5A}
	\end{figure}

	In this decomposition, Lingas, Wasylewicz, and \.Zyli\'nski~\cite{3approx} observe that the visibility region of a point guard intersects at most five histograms. 
	In  \cref{fig:5A}, the visibility region of a guard $g$ that intersects five histograms~is highlighted in blue.
	We observe that this property equally applies to \rguards{0} in polyominoes, as (1) polyominoes are orthogonal polygons, and (2) \rguards{0} (and therefore \rguards{k}) see no more than a point guard at the same position. 
	We also note that (as argued in the proof of~\cref{thm:k-guard-poly}) all histograms resulting from Levcopoulos's partition are embedded histograms. 

	Let $\apx$ be an $\alpha$-approximation that computes a \rguard{k} set for a given embedded histogram. 
	We obtain a $5\alpha$-approximation $\mathcal{A}'$ for hole-free polyominoes as follows.
	By~\cref{lem:partition}, we obtain a partition of $P$ into a set of $\ell$ embedded histograms $\mathcal{H}=\{H_1, \dots, H_{\ell}\}$. 
	We then run  $\apx$ on every $H_i \in \mathcal{H}$ and return the union of all derived \rguard{k} sets. 

	Let $G$ be a minimum guard set for~$P$ and $\opt(X)$ a minimum set of guards in $P$ for guarding the subset $X\subseteq P$.
	Moreover, let  $G_i$ denote the set of guards in $G$ that help to guard~$H_i$, i.e., see at least one cell of $H_i$. 
	By the above discussion, each guard appears in at most 5 $G_i$'s, i.e., ${\sum_{i}|G_i|\leq 5\cdot|G| = 5\cdot\opt(P)}$. 
	Moreover, $G_i$ is a valid guard set of $H_i$ and hence, $\opt(H_i)\leq |G_i|$. 
	Together, this yields the following chain of inequalities: 
	\[
	\sum_{H_i \in \mathcal{H}} |\apx(H_i)| 
	\; \leq\; 
	\sum_{H_i \in \mathcal{H}} \alpha \cdot\opt(H_i)
	\; \leq\; 
	\sum_{H_i\in \mathcal{H}} \alpha\cdot|G_i|
	\; \leq\; 
	5\alpha \cdot \opt(P)
	\]
	This completes the proof. 
\end{proof}


\section{Omitted details for~\cref{sec:better-approx-fixed-k}}
\label{app:hole-free-1-restless}

We now fix $k=1$ and give all details for proving \cref{thm:1-guard-poly}.

\hypertarget{proof:thmBetterAlgorithmForHoleFreeKOne}{}
\thmBetterAlgorithmForHoleFreeKOne*

\begin{proof}
	We refine the proof of  \cref{thm:k-guard-poly} by using the fact that $k=1$. 
    As before, we use \cref{lem:partition} to partition the polyomino $P$ into embedded histograms,
    and apply \cref{lem:k-guard-embedded} to approximate a minimum guard set on every embedded histogram. 
    Let~$\opt$ denote the number of guards in an optimal solution.
    The main insight is that for $k=1$, we can easily improve the maximum number of placed guards from $4\cdot \opt$ (for arbitrary~$k$) to $2\cdot \opt + 1$ for $k=1$. 
    The approximation factor then depends on the worst ratio between $\opt$ and $2\cdot \opt + 1$, which is 3 for the case $\opt=1$. 
    However, by checking all possible combinations of guard placements up to a fixed cardinality and preferring such optimal guardings over the approximation, we can decrease the approximation factor as close as desired to factor 2, while retaining polynominal runtime. 
    Because the generalization from embedded histograms to hole-free polyominoes introduces a factor~5 by \cref{lem:5}, we obtain a ($10+\varepsilon$)-approximation algorithm for every $\varepsilon > 0$. 
    
    It remains to show that an embedded histogram $H$ can be \guarded{1} with at most $2\cdot\opt+1$ guards, in polynomial time. 
    As in the proof of \cref{lem:k-guard-embedded}, we partition $H$ into the sets $T$, $B$, and $R$, see also \cref{fig:approx-1-improved}.
        \begin{figure}[htb]
        \centering
        \includegraphics[]{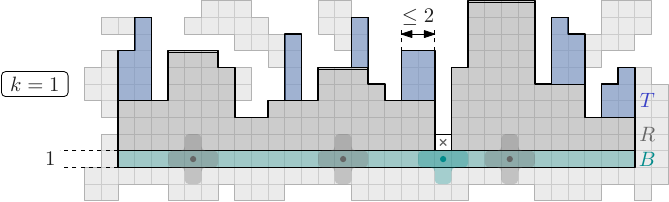}
        \caption{
            Illustration for the proof of \cref{thm:1-guard-poly} -- a simplification of our techniques for $k=1$. 
            We use that the base $B$ is a single strip and $T$ consists of embedded histograms of width at most 2.
        }
        \label{fig:approx-1-improved}
    \end{figure}
    While we use the fact that $R$ can be \guarded{1} with at most $\opt$ guards as before, we present improvements for the sets $B$ and $T$. 
    In particular, we show  that $B$ is always guardable by a single \rguard{1}, and that $T$ can be guarded optimally, i.e., with at most $\opt$ guards. 
    In total, we thus place $(2\cdot\opt+1)$ \rguards{1} which together guard $H$. 
    Let $n$ denote the number of vertices of $P$.
    
    \begin{claim}
        \label{clm:single-guard}
        We can identify a single \rguard{1} which sees $B$ in time linear in $n$.
    \end{claim}
    \begin{claimproof}
        Note that for $k=1$, $B$ is exactly the histogram's base. 
        Let $H_2$ be the set of all cells one position above $B$. 
        It is easy to see that if~$H_2$ contains no outer cell, placing a \rguard{1} anywhere on $B$ guards~$B$ entirely.
        If $H_2$ contains an outer cell (as in \cref{fig:approx-1-improved}), placing a \rguard{1} on $B$ right below this outer cell guards $B$ entirely. 
    \end{claimproof}

    \newcommand{\thinhist}[0]{\ensuremath{H}\xspace}
    \newcommand{\thinhistPI}[0]{\ensuremath{\Pi}\xspace}
    \newcommand{\thinhistI}[0]{\text{I}\xspace}

It remains to show that an optimal \guard{1} set of $T$ can be computed in time polynomial in $n$.
Note that for $k=1$, $T$ consists only of embedded histograms of width 1 or~2. 
Hence it suffices to  show in the following \cref{clm:thin} that we can optimally \guard{1} each of these histograms in polynomial time. 
By \cref{lem:separation-better}, all these histograms require separate guards, allowing us to combine all \rguards{k} for the single embedded histograms to optimally guard their union $T$. 

	To guard embedded histograms of width 1 or 2, we start with some definitions. 
    For an embedded histogram of width 2, let $\thinhistPI$ denote the union of all rows of width 2 and  let \thinhistI be $\thinhist \setminus \thinhistPI$. 
    We know that $\thinhistPI$ is adjacent to at least one outer cell~$o$ to its top; 
    without loss of generality, we assume $o$ to be at the left side (of both possible positions) in the following proofs. 
    We illustrate this setting in \cref{fig:approx-1-details}.  

    \begin{figure}[htb]
        \centering
        \includegraphics[]{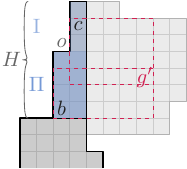}
        \caption{
            Terminology and details on guarding an embedded histogram $\thinhist$ of width 2. 
            We note that a guard $g'$ that sees both cell $b$ and $c$ also sees all cells in between, including all cells in $\thinhistPI$.}
        \label{fig:approx-1-details}
    \end{figure}

    \begin{claim}
        \label{clm:two-and-one}
        We can identify a \rguard{1} in $P$ that sees both $\thinhistPI$ and the longest possible substrip of \thinhistI from its bottom cell in time polynomial in $n$. 
    \end{claim}
    \begin{claimproof}
        We first show that we can always place a \rguard{1} that sees $\thinhistPI$.
        We find such a position on the vertical strip of cells in $\thinhistPI$ below $o$; 
        if $\thinhistPI$ is adjacent to an outer cell to its left, we place the guard next to it, otherwise we place it directly below $o$. 
        Note that this also works if~$\thinhistPI$ has a height of 1 or 2 because $\thinhist$ has an \fullbase or \emptybase.

        If $\thinhistI = \emptyset$, this already identifies a suitable guard. 
        If $\thinhistI \neq \emptyset$, we denote the bottom left cell in $\thinhistPI$ by $b$. 
        Let $c$ be a cell in~$\thinhistI$. 
        We show that a \rguard{1} $g$ sees both $b$ and $c$ exactly if it sees~$\thinhistPI$ and all cells in $\thinhistI$ not above $c$. 

        First, if $g$ sees $\thinhistPI$ and all cells in $\thinhistI$ not above $c$, this includes $b$ and $c$. 
        Now, in the case where $g$ sees $b$ and $c$, let $g'$ be a \rguard{0} in the 1-neighborhood of $g$. 
        The outer cell~$o$ restricts the possible placements of $g'$ (that sees both $b$ and $c$); $g'$ must be both to the right and below $o$. 
        Now, combining that $\rect(g',b)$, $\rect(g',c)$, and~$\thinhistPI$ are contained in $P$, it follows that $g'$ sees $\thinhistPI$ and all cells in $\thinhistI$ below $c$, as visualized in \cref{fig:approx-1-details}.

        With this statement, we can now use \cref{lem:guardable} and perform a (binary) search to find the maximum number of cells in $\thinhistI$ that can be guarded together with $\thinhistPI$. 
        As also argued in \cref{lem:strip-guarding-interval-cover}, we know by \cref{prop:existencePolytime} that we only need to search among $\bO(n)$ interesting cells. 
        This concludes the proof. 
    \end{claimproof}
    
    We use \Cref{clm:two-and-one} to show the following.

    \begin{claim}
        \label{clm:thin}
        For an embedded histogram \thinhist of width at most 2, an optimal \guard{1} set can be computed in time polynomial in $n$. 
    \end{claim}
    \begin{claimproof}
        Clearly, if \thinhist has width 1, it can be optimally guarded using the interval cover strategy from \cref{lem:strip-guarding-interval-cover}.
        If \thinhist has width 2, we apply \cref{clm:two-and-one} and identify a guard $\pi$ that sees~$\thinhistPI$ and the largest number of cells from $\thinhistI$ starting from the bottom. 
        Let $\thinhistI' \subseteq \thinhistI$ be the cells not yet guarded. 
        We guard $\thinhistI'$ using the interval cover strategy from \cref{lem:strip-guarding-interval-cover} and derive the guard set $G'$. 

        We show that $G' \cup \{\pi\}$ is an optimal guard set for $\thinhist$. 
        Otherwise there exists a guard set~$G$ of  $\thinhist$ with $|G| \leq |G'|$. 
        Because $\thinhistI'$ needs $|G'|$ guards, every guard in~$G$  sees some private cell in $\thinhistI'\subseteq \thinhistI$. 
        Note that a guard $g$ seeing a (private) cell $c\in \thinhistI$ and a cell $d\in\thinhistPI$, also sees every cell~$c'$ in~$\thinhistI$ below~$c$: Either $g$ lies below $o$ and hence $\rect(g,c)$ contains $c'$ or $\rect(g,d)$ contains~$c'$. 
        Consequently, there exists a unique guard $\bar g$ in $G$ that sees a cell of~$\thinhistPI$, otherwise~$G$ contains guards without a private cell in~$\thinhistI'$. 
        In particular, $\bar g\in G$ sees all of $\thinhistPI$ and some cells in~$\thinhistI$, such that $|G'|-1$ guards suffice to cover the remainder $\thinhistI''$ of  $\thinhistI$. By our choice of $\pi$, $\thinhistI'\subseteq \thinhistI''$. Hence, $|G'|-1$ guards suffice for $\thinhistI'$, a contradiction. 
    \end{claimproof}
    
    With these improvements, embedded histograms can be guarded with at most $2\cdot\opt + 1$ \rguards{1}, which can be turned into a $(10+\varepsilon$)-approximation for hole-free polyominoes, as described at the beginning of the proof.
\end{proof}
\end{document}